\documentclass[12pt]{amsart}

\usepackage{amssymb,amsfonts,bm}
\usepackage{tikz}
\usepackage{color}
\usepackage{graphicx}
\usepackage{stmaryrd}		
\usepackage{hyperref}
\DeclareMathOperator{\rep}{rep}
\DeclareMathOperator{\val}{val}
\DeclareMathOperator{\ord}{ord}
\theoremstyle{plain}
\newtheorem{theorem}{Theorem}
\newtheorem{lemma}[theorem]{Lemma}
\newtheorem{corollary}[theorem]{Corollary}
\newtheorem{prop}[theorem]{Proposition}
\newtheorem{conjecture}[theorem]{Conjecture}
\theoremstyle{definition}
\newtheorem{definition}[theorem]{Definition}
\newtheorem{example}[theorem]{Example}
\newtheorem{remark}[theorem]{Remark} 
\newtheorem*{notation*}{Notation}
\newtheorem*{remark*}{Remark}
\numberwithin{equation}{section}
\newtheorem{problem}{Problem}

\newcommand{\N}{\mathbb{N}}
\newcommand{\Q}{\mathbb{Q}}
\newcommand{\Z}{\mathbb{Z}}

\newcommand{\nequiv}{\mathrel{\not\equiv}}
\newcommand{\colonequal}{\mathrel{\mathop:}=}

\newcommand{\size}[1]{\lvert{#1}\rvert}
\newcommand{\sizedsize}[1]{\left\lvert{#1}\right\rvert}

\title[Ultimate periodicity problem for numeration systems]{Ultimate periodicity problem for linear numeration systems}
\author[\'E. Charlier, A. Massuir, M. Rigo]{\'E. Charlier$^1$, A. Massuir$^1$, M. Rigo$^1$}
\address{$^1$ University of Li\`ege, Department of Mathematics, All\'ee de la d\'ecouverte 12 (B37), B-4000 Li\`ege, Belgium}
\email{echarlier@uliege.be ; A.Massuir@uliege.be ; M.Rigo@uliege.be}
\author[E. Rowland]{E. Rowland$^2$}
\address{$^2$ Department of Mathematics, Hofstra University, Hempstead, NY, USA}
\email{Eric.Rowland@hofstra.edu}

\begin{document}

\begin{abstract}
    We address the following decision problem.  Given a numeration
    system $U$ and a $U$-recognizable set $X\subseteq\mathbb{N}$,
    i.e.\ the set of its greedy $U$-representations is recognized by a
    finite automaton, decide whether or not $X$ is ultimately
    periodic. We prove that this problem is decidable for a large
    class of numeration systems built on linear recurrence 
    sequences. Based on arithmetical considerations about the
    recurrence equation and on $p$-adic methods, the DFA given as
    input provides a bound on the admissible periods to test.
\end{abstract}


\maketitle

\section{Introduction}

Let us first recall the general setting of linear numeration systems that are used to represent, in a greedy way, non-negative integers by words over a finite alphabet of digits. See, for instance, \cite{Fraenkel}. Let $\N = \{0, 1, 2, \dots\}$.

\begin{definition}\label{def:numsys}
    A {\it numeration system} is given by an increasing
    sequence $U=(U_i)_{i\ge 0}$ of integers such that $U_0=1$ and
    $C_U:=\sup_{i\ge 0} \lceil \frac{U_{i+1}}{U_i}\rceil$ is finite. Let
    $A_U=\{0,\ldots,C_U-1\}$ be the canonical alphabet of digits. The {\it greedy $U$-representation} of a
    positive integer $n$ is the unique finite word
    $\rep_U(n)=w_\ell\cdots w_0$ over $A_U$ satisfying
$$n=\sum_{i=0}^\ell w_i\, U_i,\ w_\ell\neq 0 \text{ and } \sum_{i=0}^t w_i\,
U_i<U_{t+1},\ t=0,\ldots,\ell.$$ We set $\rep_U(0)$ to be the
empty word $\varepsilon$. A set $X\subseteq\mathbb{N}$ of integers is
{\it $U$-recognizable} if the language $\rep_U(X)$ over $A_U$ is
regular (i.e.\  accepted by a finite automaton). 
\end{definition}

Recognizable sets of integers are considered as particularly simple because membership can be decided  by a deterministic finite automaton in linear time with respect to the length of the representation. It is well known that such a property for a subset of $\mathbb{N}$ depends on the choice of the numeration system. For a survey on integer base systems, see \cite{BHMV}. For generalized numeration systems, see \cite{FLANS}. {For basic results in automata theory, see, for instance \cite{FLANS,Sakarovitch}.}

\begin{definition}
If $x=x_\ell\cdots
x_0$ is a word over an alphabet of 
integers, then
the {\it $U$-numerical value} of $x$ is $$\val_U(x)=\sum_{i=0}^\ell
x_i\, U_i.$$
\end{definition}

From the point of view of formal languages, it is quite desirable that $\rep_U(\mathbb{N})$ is regular; we want to be able to check whether or not a word is a valid greedy $U$-representation. This implies that $U$ satisfies a linear recurrence relation. See, for instance, \cite{Sha} or \cite[Prop.~3.1.5]{cant}.

\begin{definition}\label{def:linsys}
    A numeration system $U$ is said to be {\it linear} if it ultimately satisfies a homogeneous linear recurrence relation with integer coefficients. There exist $k\ge 1$, $a_{k-1},\ldots,a_0\in\mathbb{Z}$ such that $a_0\neq 0$ and $N\ge 0$ such that for all $i\ge N$, 
    \begin{equation}
        \label{eq:linrec}
        U_{i+k}=a_{k-1} U_{i+k-1}+\cdots +a_0 U_i.
    \end{equation}
The polynomial $X^N(X^k-a_{k-1}X^{k-1}-\cdots -a_0)$ is called the {\em characteristic polynomial} of the system {(where it is assumed that $k$ and then $N$ are chosen to be minimal)}.
\end{definition}

The regularity of $\rep_U(\mathbb{N})$ is also important for another reason. The language  $\rep_U(\mathbb{N})$ is regular if and only if every ultimately periodic set of integers is $U$-recognizable \cite[Thm.~4]{LR}. In particular, as recalled in Proposition~\ref{pro:up}, if an ultimately periodic set $X$ is given, then a DFA accepting $\rep_U(X)$ can effectively be obtained.

In this paper, we address the following decidability question. Our aim is to prove that this problem is decidable for a large class of numeration systems.

\begin{problem}\label{pb}
    Given a linear numeration system $U$ and a
    (deterministic) finite automaton~$\mathcal{A}$ whose accepted language is contained in the numeration language $\rep_U(\mathbb{N})$, decide whether 
   the subset $X$ of $\mathbb{N}$ that is recognized by $\mathcal{A}$ is ultimately periodic, i.e.\  whether or not $X$ is a finite
    union of arithmetic progressions (along a finite set). 
\end{problem}

This question about ultimately periodic sets is motivated by the celebrated theorem of Cobham. Let $p,q\ge 2$ be integers. If $p$ and $q$ are multiplicatively independent, i.e.\  $\frac{\log(p)}{\log(q)}$ is irrational, then the ultimately periodic sets are the only sets that are both $p$-recognizable and $q$-recognizable \cite{cobham}. These are exactly the sets definable by a first-order formula in the Presburger arithmetic $\langle\mathbb{N},+\rangle$. Cobham's result has been extended to various settings; see \cite{Durand,Mitrofanov} for an application to morphic words. See \cite{DurRig} for a survey.

In this paper, we write greedy $U$-representations with most significant digit first (MSDF convention): the leftmost digit is associated with the largest $U_\ell$ occurring in the decomposition. Considering least significant digit first would not affect decidability (a language is regular if and only if its reversal is) but this could have some importance in terms of complexity issues not discussed here.

\medskip

\textbf{What is known.} Let us quickly review cases where the decision problem is known to be decidable. Relying on number theoretic results, the problem was first solved by Honkala for integer base systems \cite{Honkala}. An alternative approach bounding the syntactic
complexity of ultimately periodic sets of integers written in base $b$ was studied in \cite{LRRV}. 
Recently a deep analysis of the structure of the automata accepting ultimately periodic sets has led to an efficient decision procedure for integer base systems \cite{MSaka,BMMR,Marsault2019}. An integer base system is a particular case of a Pisot system, i.e.\  a linear numeration system whose characteristic polynomial is the minimal polynomial of a Pisot number (an algebraic integer larger than~$1$ whose conjugates all have modulus less than one). For these systems, one can make use of first-order logic and the decidable extension $\langle \mathbb{N},+,V_U\rangle$ of Presburger arithmetic \cite{BH}. For an integer base~$p$, $V_p(n)$ is the largest power of $p$ dividing $n$. A typical example of Pisot system is given by the Zeckendorf system based on the Fibonacci sequence $1,2,3,5,8,\ldots$. Given a $U$-recognizable set $X$, there exists a first-order formula $\varphi(n)$ in $\langle \mathbb{N},+,V_U\rangle$ describing $X$. The formula $$(\exists N)(\exists p)(\forall n\ge N)(\varphi(n)\Leftrightarrow \varphi(n+p))$$ thus expresses when $X$ is ultimately periodic, $N$ being a preperiod and $p$ a period of $X$. The logical formalism can be applied to systems such that the addition is $U$-recognizable by an automaton, i.e.\  the set $\{(x,y,z)\in\mathbb{N}^3:x+y=z\}$ is $U$-recognizable. This is the case for Pisot systems \cite{Frou}.

When addition is not known to be $U$-recognizable, other techniques must be sought. Hence the problem was also shown to be decidable for some non-Pisot linear numeration systems satisfying a gap condition $\lim_{i\to+\infty} U_{i+1}-U_i=+\infty$ and a more technical condition $\lim_{m\to+\infty} \mathcal{N}_U(m)=+\infty$ where $\mathcal{N}_U(m)$ is the number of  residue classes that appear infinitely often in the sequence $(U_i\bmod{m})_{i\ge 0}$; see  \cite{BCFR}. An example of such a system is built on the relation $U_i=3\, U_{i-1}+2\, U_{i-2}+3\, U_{i-3}$ \cite{Frougny97}. For extra pointers to the literature (such as an extension to a multidimensional setting), the reader can follow the introduction in \cite{BCFR}.

\medskip

\textbf{Our contribution.} In view of the above summary, we are looking for a decision procedure that may be applied to non-Pisot linear numeration systems such that $\mathcal{N}_U(m)\not\to \infty$ when $m$ tends to infinity. Hence we want to take into account systems where we are not able to apply a decision procedure based on first-order logic nor on the technique from \cite{BCFR}. We follow Honkala's original scheme: if a DFA~$\mathcal{A}$ is given as input (the question being whether the corresponding recognized subset of $\mathbb{N}$ is ultimately periodic), the number of states of $\mathcal{A}$ should provide an upper bound on the admissible preperiods~$s$ and periods~$p$. If there is a finite number of such pairs to test, then we build a DFA $\mathcal{A}_{s,p}$ for each pair $(s,p)$ and one can test whether or not the two automata $\mathcal{A}$ and $\mathcal{A}_{s,p}$ accept the same language. This provides us with a decision procedure. Roughly speaking, if the given DFA is ``small'', then it cannot accept an ultimately periodic set with a minimal period being ``overly complicated'', i.e.\  ``quite large''.

\begin{example}\label{exa:ppp}
    Here is an example of a numeration system based on a Parry (the $\beta$-expansion of $1$ is finite or ultimately periodic, see \cite[Chap.~2]{cant}) non-Pisot number $\beta$: $$U_{i+4}=2\, U_{i+3}+2\, U_{i+2}+2\, U_i.$$
Indeed, the largest root $\beta$ of the characteristic polynomial is roughly $2.804$, and $-1.134$ is another root of modulus larger than one. With the initial conditions $1,3,9,23$, $\rep_U(\mathbb{N})$ is the regular language over $\{0,1,2\}$ of words avoiding factors $2202$, $221$ and $222$. For details, see \cite[Ex.~2.3.37]{cant} or \cite{PhD}. When $m$ is a power of $2$, there is a unique congruence class visited infinitely often by the sequence $(U_i\bmod{m})_{i\ge 0}$ because $U_i\equiv 0\pmod{2^r}$ for large enough $i$. For such an example, $\mathcal{N}_U(m)$ does not tend to infinity and thus the previously known decision procedures may not be applied. This is a perfect candidate for which no decision procedures are known.
\end{example}

This paper is organized as follows. In Section~\ref{sec2}, we make clear our assumptions on the numeration system. In Section~\ref{sec3}, we collect several known results on periodic sets and $U$-representations. In particular, we relate the length of the $U$-representation of an integer to its value. The core of the paper is made of Section~\ref{sec4} where we discuss cases to bound the admissible periods. In particular, we consider two kinds of prime factors of the admissible periods: those that divide all the coefficients of the recurrence and those that don't, see \eqref{eq:fQ}. In Section~\ref{sec5}, we apply the discussion of the previous section. First, we obtain a decision procedure when the gcd of the coefficients of the recurrence relation is~$1$, see Theorem~\ref{thm:main1}. This extends the scope of results from \cite{BCFR}. On the other hand, if there exist primes dividing all the coefficients, our approach heavily relies on quite general arithmetic properties of linear recurrence relations. It has therefore inherent limitations because of notoriously difficult problems in $p$-adic analysis such as finding bounds on the growth rate of blocks of zeroes in the digit sequences of $p$-adic numbers of a special logarithmic form. We discuss the question and give illustrations of these $p$-adic techniques in Section~\ref{sec:ap}. The paper ends with some concluding remarks.


\section{Our setting}\label{sec2}
We have minimal assumptions on the considered linear numeration system $U$. 
\begin{itemize}
  \item[(H1)] $\mathbb{N}$ is $U$-recognizable.
  \item[(H2)] There are arbitrarily large gaps between consecutive terms:
$$\limsup_{i\to+\infty}(U_{i+1}-U_i)=+\infty.$$
  \item[(H3)] The gap sequence $(U_{i+1}-U_i)_{i\ge 0}$ is ultimately non-decreasing: there exists {$G\ge 0$ such that for all $i\ge G$}, 
$$U_{i+1}-U_i \le U_{i+2}-U_{i+1}.$$
\end{itemize}

Note that Example~\ref{exa:ppp} satisfies all the above assumptions. Let us make a few comments.

\begin{remark}
(H1) gives sense and meaning to our decision problem; under that assumption, ultimately periodic sets are $U$-recognizable. As recalled in the introduction, it is a well known result that (H1) implies that the numeration system $(U_i)_{i\ge 0}$ satisfies a linear recurrence relation with integer coefficients as in \eqref{eq:linrec}. { Some sufficient conditions that guarantee $\mathbb{N}$ to be $U$-recognizable are given in \cite{Lo,Ho}. However, the general case remains open, see \cite[Section~8.2]{Ho}.}
\end{remark}

\begin{remark}\label{rem:h3w}
The assumptions (H2) and (H3) imply that $\lim_{i\to+\infty}(U_{i+1}-U_i)=+\infty$. However, in many cases, even if $\lim_{i\to+\infty}(U_{i+1}-U_i)=+\infty$, the gap sequence may decrease from time to time. So, even a stronger assumption than (H2) does not imply (H3). The main reason why we introduce (H3) is the following one. Let $1 0^\ell w$ be a greedy $U$-representation for some $\ell\ge 0$. Assume (H3) and $i=|w|+\ell\ge G$. Then for all $\ell'\ge \ell$,  $10^{\ell'} w$ is a greedy $U$-representation as well. Indeed, if $n$ is a non-negative integer such that $U_i+n<U_{i+1}$, then $U_{i+1}+n=U_{i+1}-U_i+U_i+n\le U_{i+2}-U_{i+1}+U_i+n<U_{i+2}$. Hence $U_{i'}+n<U_{i'+1}$ for all $i'\ge i$, meaning that as soon as the greediness property is fulfilled, one can shift the leading $1$ at every larger index. This is not always the case, as seen in Example~\ref{exa:noth2}. This property will be used in Lemma~\ref{lem:h3}, which in turn will be crucial in the proofs of Propositions~\ref{pro:cas2a} as well as Theorem~\ref{the:main}, where we construct $U$-representations with leading $1$'s in convenient positions. 
\end{remark}

\begin{remark}\label{rem:hard}
  {
    If $(U_i)_{i\ge 0}$ is a linear recurrence sequence, so are the first and second differences $(V_i)_{i\ge 0}:=(U_{i+1}-U_i)_{i\ge 0}$ and $(W_i)_{i\ge 0}:=(V_{i+1}-V_i)_{i\ge 0}$. (H3) can be restated as follows. There exists $G$ such that $W_i\ge 0$ for all $i\ge G$. It relates to the {\em Ultimate Positivity Problem}: given a linear recurrence sequence $(W_i)_{i\ge 0}$, are all but finitely many terms of $(W)_{i\ge 0}$ non-negative? This problem is known to be decidable for integer linear recurrence sequences of order at most $5$ in polynomial time \cite{Ou}. It is also decidable whenever the characteristic polynomial only has simple roots \cite{Ou2}. However, in a general setting, it remains a longstanding open problem \cite{SS}.
    }
\end{remark}

\begin{remark}{
The following deep result due independently to Evertse and to van der Poorten and Schlickewei is discussed in \cite{Ou3}, see the terminology and the references therein:  For any non-degenerate algebraic linear recurrence sequence $(V_i)_{i\ge 0}$ of dominant modulus $\rho > 1$, and
  any $\varepsilon > 0$, there exists a constant $H$ such that, for all $i\ge H$, we have $|V_i|\ge\rho^{(1-\varepsilon)i}$. As noticed in \cite{Ou2}, any degenerate linear recurrence sequence can be effectively decomposed into a finite number of non-degenerate linear recurrence sequences.  In our setting of numeration systems, the sequence $(U_i)_{i\ge 0}$ is increasing so the gap sequence $(V_i)_{i\ge 0}:=(U_{i+1}-U_i)_{i\ge 0}$ is positive and (H2) is thus satisfied whenever the associated dominant modulus is larger than $1$. The presence of a root of modulus larger than $1$ can be tested with the Lehmer--Schur algorithm, see \cite[Chap.~10]{Morris}.}
\end{remark}

\begin{example}
\label{exa:toy}
Our toy example that will be treated all along the paper is given by the recurrence $U_{i+3}=12\, U_{i+2}+6\, U_{i+1}+12\, U_i$. Even though the system is associated with a Pisot number, it is still interesting because $\mathcal{N}_U(m)$ does not tend to infinity (so we cannot follow the decision procedure from \cite{BCFR}) and the gcd of the coefficients of the recurrence is larger than~$1$. Let $r\ge 1$. If the modulus is a power of $2$ or $3$, then $U_i\equiv 0\pmod{2^r}$ (resp.\ $U_i\equiv 0\pmod{3^r}$) for large enough $i$. By taking the initial conditions $1,13,163$, the language of greedy $U$-representations is regular. For the reader aware of $\beta$-numeration systems, let us mention that this choice of initial conditions corresponds to the Bertrand initial conditions, in which case the language $\rep_U(\mathbb{N})$ is equal to the set of factors (with no leading zeroes) occurring in the $\beta$-expansions of real numbers where $\beta$ is the dominant root of the characteristic polynomial $X^3-12X^2-6X-12$ of the recurrence relation of the system $U$ \cite{Bertrand-Mathis:1989}. 
\end{example}


\section{Some classical lemmas}\label{sec3}

A set $X\subseteq\mathbb{N}$ is {\em ultimately periodic} if its characteristic sequence $\mathbf{1}_X\in\{0,1\}^\mathbb{N}$ is of the form $uv^\omega$ where $u,v$ are two finite words over $\{0,1\}$. It is assumed that $u,v$ are chosen of minimal length. Hence the {\em period} of $X$ denoted by $\pi_X$ is the length $|v|$ and its preperiod is the length $|u|$. We say that $X$ is {\em (purely) periodic} whenever the preperiod is zero. The following lemma is a simple consequence of the minimality of the
period chosen to represent an ultimately periodic set.

\begin{lemma}\label{lem:per}
    Let $X\subseteq\mathbb{N}$ be an ultimately periodic set of period $\pi_X$ and let $i,j$ be integers greater than or equal to the preperiod of $X$. If $i\not\equiv j\pmod{
\pi_X}$ then there exists $r<\pi_X$ such that either $i+r\in X$ and
    $j+r\not\in X$ or, $i+r\not\in X$ and $j+r\in X$.
\end{lemma}

Our assumption (H2) permits us to extend greedy $U$-representations with some extra leading digits. See \cite[Lemma~7]{BCFR} for a proof.

\begin{lemma}\label{lem:condlim}
    Let $U$ be a numeration system satisfying (H2).
For all greedy $U$-representations $w$, there exists arbitrarily large $r$ such that the word $10^{r}w$ is also a greedy $U$-representation.
\end{lemma}

When $\mathbb{N}$ is $U$-recognizable, using a pumping-like argument, we can give an upper bound on the number of zeroes to be inserted.

{
\begin{lemma}\label{lem:h12}
Let $U$ be a numeration system satisfying (H1) and (H2). Then there is an integer constant $C >0$ such that if $w$ is a greedy $U$-representation, then for some $\ell < C$, $10^{\ell}w$ is also a greedy $U$-representation.
\end{lemma}
\begin{proof}
By assumption (H1), there is a DFA, say with $C$ states, accepting the numeration language $\rep_U (\N)$. Let $w$ be a greedy $U$-representation. Then from Lemma~\ref{lem:condlim}, there is $r \geq C$ such that $10^rw \in \rep_U(\N)$. The path of label $10^rw$ starting from the initial state is accepting. Since $r \geq C$, a state is visited at least twice when reading the block $0^r$. Thus there is an accepting path of label $10^{\ell}w$ with $\ell < C$.
\end{proof}

Let us introduce a constant $Z$.

\begin{definition}\label{def:Z}
Let $U$ be a numeration system satisfying (H1), (H2) and (H3). Thanks to (H2), there exist infinitely many $R$ such that
\[
U_{R+1} - U_R \geq U_{i+1} - U_i
\]
for all $i \leq  R$. We may choose the least $R$ with this property and such that $R\ge G$ where $G$ is the constant given in (H3). We set 
\[
Z = \max \{ R, C \}
\]
where $C$ is the constant given in Lemma~\ref{lem:h12}.
\end{definition}
}

{
In view of Remark~\ref{rem:hard} about the status of the general decision problem about (H3), we assume that $G$ is given as an input with the numeration system. Hence the constants $C,R,Z$ can be effectively computed. Indeed, $C$ can be deduced from the automaton accepting the language of the numeration. Then $R$ can be computed by an exhaustive search and finally, one has to choose $Z = \max \{ R,C\}$.
}

{
Thanks to (H3), we have more flexibility about the inserted zeroes: we can add as many zeroes as needed to greedy representations and obtain again greedy representations.

\begin{lemma}\label{lem:h3}
Let $U$ be a numeration system satisfying (H1), (H2) and (H3). If $w$ is a greedy $U$-representation, then for all $z \geq Z$, $10^zw$ is also a greedy $U$-representation.
\end{lemma}
\begin{proof}
Let $w$ be a greedy $U$-representation. By Lemma~\ref{lem:h12}, there is $\ell < C$ such that $10^{\ell}w$ is a greedy $U$-representation. Let $i = \ell + \vert w \vert$. Let $n = \val_U(w)$. We have $U_i + n < U_{i+1}$.
\begin{itemize}
\item If $i \geq Z$, similarly as in Remark~\ref{rem:h3w}, since $Z\ge G$, 
$$U_{i+1} + n = U_{i+1} - U_i + U_i + n \leq U_{i+2} - U_{i+1} + U_i + n < U_{i+2}.$$
Hence $U_j + n < U_{j+1}$ for all $j \geq i$. Otherwise stated, $10^{\ell^{\prime}}w$ is a greedy $U$-representation for all $\ell^{\prime} \geq \ell$. In particular, since $\ell < Z$, for all $z \geq Z$, $10^zw$ is a greedy $U$-representation.
\item If $i < Z$, then
\begin{align*}
U_{Z} + n &= U_{Z} - U_i + U_i + n < U_{Z} - U_{i} + U_{i+1}\\
&\leq U_{Z+1} - U_{i+1} + U_{i+1} \leq U_{Z+1}
\end{align*}
because $Z\ge R\ge G$. Hence $10^{Z {-} \vert w \vert}w$ is a greedy $U$-representation. We conclude by applying the first part of the proof.
\end{itemize}
\end{proof}

}

\begin{example}\label{exa:noth2}
	The sequence $1,2,4,5,16,17,64,65,\ldots$ is a solution of the linear recurrence $U_{i+4}=5U_{i+2}-4U_i$ but it does not satisfy (H3). The property stated in Lemma~\ref{lem:h3} does not hold: only some shifts to the left of the leading coefficient $1$ lead to valid greedy expansions. The word $1001$ is the greedy representation of $6$ but for all $t\ge 1$, $1(00)^t1001$ is not a greedy representation.
\end{example}

\begin{example}\label{exa:noth2-newh2ok}
    The sequence $1,2,3,4,8,12,16,32,48,64,128,\ldots$ is a solution of the linear recurrence $U_{i+3}=4U_i$. The numeration language $0^*\rep_U(\mathbb{N})$ is the set of suffixes of $\{000,001,010,100\}^*$, hence (H1) holds. For all $i\ge 0$, $U_{i+1}-U_i=4^{\lfloor i/3\rfloor}$. Therefore, (H2) and (H3) are also verified.
\end{example}

We will also make use of the following folklore result. See, for instance, \cite[Prop.~3.1.9]{cant}. It relies on the fact that a linear recurrence sequence is ultimately periodic modulo~$Q$. 

\begin{prop}
\label{pro:up}
    Let $Q,r\ge 0$. Let $A\subseteq\mathbb{N}$ be a finite alphabet. 
If $U$ is a linear numeration system, then 
$$\left\{w\in A^*\mid 
\val_U(w)\in Q\, \mathbb{N}+r\right\}$$
is accepted by a DFA that can be effectively constructed. In
particular, whenever $\mathbb{N}$ is $U$-recognizable, i.e.\  under (H1), then any ultimately periodic
set is $U$-recognizable.
\end{prop}

Under assumption (H1) the formal series $\sum_{i\ge 0}U_i\, X^i$ is $\mathbb{N}$-rational because $U_i$ is the number of words of length less than or equal to $i$ in the regular language $\rep_U(\mathbb{N})$. One can therefore make use of Soittola's theorem \cite[Thm.~10.2]{SS}: The series is the merge of rational series with dominating eigenvalues and polynomials. We thus define the following quantities.

\begin{definition}\label{def:uT}
  We introduce an integer $u$ and a real number ${\beta}$ depending only on the numeration system. From Soittola's theorem, there exist an integer $u\ge 1$, real numbers $\beta_0,\ldots,\beta_{u-1}\ge 1$ and non-zero polynomials $P_0,\ldots,P_{u-1}$ such that for $r\in\{0,\ldots,u-1\}$ and large enough $i$, {say $i\ge I_1$,}
  \begin{equation}
    \label{eq:soit}
    U_{ui+r}= P_r(i)\, \beta_r^i +Q_r(i)
  \end{equation}
where $\frac{Q_r(i)}{\beta_r^i}\to 0$ when $i\to\infty$. Since $(U_i)_{i\ge 0}$ is increasing, for $r<s<u$, for all $i\ge I_1$, we have
$$U_{ui+r}<U_{ui+s}<U_{u(i+1)+r}.$$
By letting $i$ tend to infinity, this shows that we must have $\beta_0=\cdots=\beta_{u-1}$ which we denote by $\beta$ and $\deg(P_0)=\cdots=\deg(P_{u-1})$ which we denote by~$d$. Otherwise stated, $U_{ui+r}\sim c_r i^{d} \beta^i$ for some constant $c_r$. Finally, let $T$ be such that $c_T=\max_{0\le r<u} c_r$. Otherwise stated, we highlight with $T$ a subsequence $(U_{ui+T})_{i\ge 0}$ with the maximal dominant coefficient.
{ Since $(U_i)_{i \in \N}$ is increasing and $\frac{Q_T(i)}{\beta^i} \rightarrow 0$ when $i \rightarrow + \infty$, there is $I_2 > 0$ such that $P_T(i) > 0$, for all $i \geq I_2$. Moreover, there is $I_3 > 0$ such that $P_T$ is non-decreasing ``after $I_3$'', i.e.\ $P_T(n)\le P_T(n+1)$ for all $n\ge I_3$. Finally, let $I$ be the positive integer $\max \{ I_1 , I_2 , I_3 \}$.
  }
\end{definition}

Note that if a numeration system has a dominant root, i.e.\  the minimal recurrence relation satisfied by $(U_i)_{i\ge 0}$ has a unique root $\beta>1$, possibly with multiplicity greater than 1, of maximum modulus, then $u=1$.

\begin{lemma}\label{lem:length}
With the notation of Definition~\ref{def:uT}, if $\beta>1$ then there exist non-negative constants $K$ and $L$ such that for all $n$, $$\size{\rep_U(n)} < u \log_{\beta}(n) + K$$
and $$\size{\rep_U(n)}  > u \log_\beta(n)- u \log_\beta(P_T(\log_\beta(n)+K/u))-L.$$
\end{lemma}

This lemma shows that the length of the greedy $U$-representation of $n$ grows at most like $\log_{\beta^{1/u}}(n)$. If $P_T$ is a constant polynomial, the lower bound is of the form $u \log_\beta(n) -L'$ for some non-negative constant~$L'$. From this result, we may express the weaker information (on ratios instead of differences) that $\size{\rep_U(n)}\sim u \log_\beta(n)$. The intricate form of the lower bound can be seen on an example such as $(U_i)_{i\ge 0}=(i^d\, 2^i)_{i\ge 0}$. In such a case, we get $\log_2(n)<\size{\rep_U(n)}+d\, \log_2(\size{\rep_U(n)})$. Hence a lower bound for $\size{\rep_U(n)}$ is less than $\log_2(n)$.

\begin{proof}
We have $\size{\rep_U(n)}=\ell$ if and only if $U_{\ell-1}\le n <U_\ell$. 
We make use of Definition~\ref{def:uT} for $u$, $\beta$, $T$ and $I$.  
Let $j=\lfloor\frac{\ell-1-T}{u}\rfloor$.

{
  Suppose that $n$ is large enough so that $j \geq I$. Since $U$ is increasing and $j \geq I$, \eqref{eq:soit} gives 
}
$$U_{\ell-1}\ge U_{uj+T} =P_T(j) \beta^j + Q_T(j).$$
We get
$$\log_{\beta}(n)
\ge \log_{\beta}(U_{\ell-1})
\ge j+\log_\beta(P_T(j))+\log_{\beta}\left(1 + \frac{Q_T(j)}{P_T(j)\beta^j}\right).$$
{ Note that, for large enough $n$, we can suppose that $1 + \frac{Q_T(j)}{P_T(j)\beta^j} >0$ (since $\frac{Q_T(i)}{\beta^i} \rightarrow 0$ when $i \rightarrow + \infty$ and $P_T$ is non-decreasing after $I$), so that the last logarithm in the above inequality is well defined.\\
Hence
\[
j \leq \log_{\beta}(n) - \log_{\beta}(P_T(j)) -\log_{\beta} \left( 1 + \frac{Q_T(j)}{P_T(j)\beta^j} \right).
\]
Moreover, $j > \frac{\ell {-} 1 {-} T}{u} {-} 1 \geq  \frac{\ell {-} u}{u} {-} 1 \geq \frac{\ell}{u} {-} 2$. We obtain
\[
\ell < u(j+2) \leq u \log_{\beta}(n) + 2u - u \log_{\beta}(P_T(j))- u \log_{\beta} \left( 1 + \frac{Q_T(j)}{P_T(j)\beta^j} \right).
\]
Since $j \geq I$ and $P_T$ is non-decreasing after $I$, we get
\[
\ell < u(j+2) \leq u \log_{\beta}(n) + 2u - u\log_{\beta}(P_T(I))- u \log_{\beta}\left( 1 + \frac{Q_T(j)}{P_T(j) \beta^j} \right). 
\]
Finally, since $\frac{Q_T(i)}{\beta^i} \rightarrow 0$ when $i \rightarrow + \infty$, there is a constant $K \geq 0$ such that
\[
\ell < u(j+2) \leq u\log_{\beta}(n) + K.
\]
We have supposed $n$ to be large enough so that $j \geq I$ and $1+\frac{Q_T(j)}{P_T(j)\beta^j}>0$. There is only a finite number of integers not fulfilling these conditions. Hence, possibly increasing the value of the constant $K$, we can assume that the above inequality holds for all integers $n$.
}

We proceed similarly to get a lower bound for $\ell$. Let $k=\lfloor\frac{\ell-T}{u}\rfloor$.
{ Observe that $j \leq k$, hence $k \geq I$. Since $U$ is increasing, we have
\[
U_{\ell} < U_{u(k+1)+T} = P_T(k+1) \beta^{k+1} + Q_T(k+1).
\]
We obtain
\[
\log_{\beta}(n) < \log_{\beta}(U_{\ell}) < k+1 + \log_{\beta} (P_T(k+1)) + \log_{\beta} \left( 1 + \frac{Q_T(k+1)}{P_T(k+1)\beta^{k+1}} \right).
\]
As in the first part of the proof, we can suppose that $n$ is large enough to get $1 + \frac{Q_T(k+1)}{P_T(k+1)\beta^{k+1}}>0$.\\
Observe that $k \leq j+1$. Hence, from the first part, we get
\[
k+1 \leq j+2 \leq \log_{\beta}(n) + \frac{K}{u}.
\]
Since $k \leq \frac{\ell {-} T}{u} \leq \frac{\ell}{u}$, we have
\[
\ell \geq uk > u \log_{\beta}(n) - u -u \log_{\beta}(P_T(k+1)) - u \log_{\beta}\left( 1 + \frac{Q_T(k+1)}{P_T(k+1)\beta^{k+1}} \right).
\]
We have $k+1 > k \geq I$ and recall that $P_T$ is non-decreasing after $I$, hence
\[
P_T(k+1) \leq P_T\left( \log_{\beta}(n) + \frac{K}{u} \right).
\]
Hence
\begin{align*}
\ell > u  \log_{\beta}(n) - u - u \log_{\beta}&\left( P_T\left( \log_{\beta}(n)+\frac{K}{u} \right) \right)\\
&- u \log_{\beta} \left( 1 + \frac{Q_T(k+1)}{P_T(k+1)\beta^{k+1}} \right).
\end{align*}
Furthermore, since $\frac{Q_T(i)}{\beta^i} \rightarrow 0$ when $i \rightarrow + \infty$ and $P_T$ is non-decreasing after $I$, there is a constant $L \geq 0$ such that
\[
\ell > u \log_{\beta}(n) - u \log_{\beta} \left( P_T \left( \log_{\beta}(n) +  \frac{K}{u} \right) \right)  - L.
\]
As in the first part of the proof, we only considered those $n$ such that $j \geq I$ and $1 + \frac{Q_T(k+1)}{P_T(k+1)\beta^{k+1}} >0$. Possibly increasing the value of $L$, we can assume that the above inequality is satisfied for all integers $n$.
  }
\end{proof}

\begin{example}
\label{exa:merge}
Consider the sequence $1,2,6,12,36,72,\ldots$ defined by $U_0=1$, $U_{2i+1}=2 U_{2i}$ and $U_{2i+2}=3U_{2i+1}$. Then for all $i\ge 0$, $U_{i+2}=6U_{i}$. It is easily seen that $U_{2i}=6^i$ and $U_{2i+1}=2\cdot 6^i$. With the notation of Definition~\ref{def:uT}, $u=2$, $\beta=6$, $d=0$ and $P_T=c_T=2$.
{
  The language $0^*\rep_U(\mathbb{N})$ is made of words where in even (resp.\ odd) positions digits belong to $0,1$ (resp.\ $0,1,2$), i.e. 
  $$0^*\rep_U(\mathbb{N})=(\varepsilon+0+1)((0+1+2)(0+1))^*.$$
}If $\size{\rep_U(n)}=2 \ell+1$ then $U_{2\ell}=6^{\ell}\le n<U_{2\ell+1}=2\cdot 6^{\ell}$, so $\size{\rep_U(n)}\le 2\log_6(n)+1$  and $\size{\rep_U(n)}>2\log_6(\frac{n}{2})+1=2\log_6(n)-2\log_6(2)+1$.
If $\size{\rep_U(n)}=2 \ell$ then $U_{2\ell-1}=2\cdot 6^{\ell-1}\le n<U_{2\ell}= 6^{\ell}$, so $\size{\rep_U(n)}\le 2\log_6(3n)= 2\log_6(n)+2\log_6(3)$ and $\size{\rep_U(n)}>2\log_6(n)$.
\end{example}

\begin{example}
Consider the sequence $1,3,8,20,48,112,\ldots$ defined by $U_0=1$, $U_1=3$ and $U_{i+2}=4U_{i+1}-4U_i$. Then $U_i=(\frac{i}{2}+1)2^i$. With the notation of Definition~\ref{def:uT}, $u=1$, $\beta=2$, $d=1$ and $P_T(n)=\frac{n}{2}+1$.
If $\size{\rep_U(n)}=\ell$ then $U_{\ell-1}=(\frac{\ell-1}{2}+1)2^{\ell-1}\le n<U_{\ell}=(\frac{\ell}{2}+1)2^{\ell}$, so $\size{\rep_U(n)}< \log_2(n)+1$ and $\size{\rep_U(n)}> \log_2(n)-\log_2(\frac{\ell}{2}+1)>\log_2(n)-\log_2(\frac{1}{2}\log_2(n)+\frac{3}{2})$. With the notation of Lemma~\ref{lem:length}, $K=1$ and $P_T(\log_2(n)+K)=\frac{1}{2}\log_2(n)+\frac{3}{2}$.
\end{example}

As shown by the next result. It is enough to obtain a bound on the possible periods of $X$. In \cite[Prop.~44]{BCFR}, the result is given in a more general setting (i.e.\  for abstract numeration systems) and we restate it in our context.

\begin{prop}\label{bound:prep}
    Let $U$ be a numeration system satisfying (H1), let $X\subseteq\mathbb{N}$ be an ultimately periodic set and let $\mathcal{A}_X$ be a DFA accepting $\rep_U(X)$. Then the preperiod of $X$ is bounded by a computable constant depending only on the size of $\mathcal{A}_X$ and the period $\pi_X$ of $X$.
\end{prop}

{
Thus, our aim is to bound the period $\pi_X$ only in terms of the given automaton recognizing $X$. 
}


\section{Number of states}\label{sec4}

We follow Honkala's strategy introduced in \cite{Honkala}. A DFA $\mathcal{A}$ accepting $\rep_U(X)$ is given as input. Assuming that $X$ is ultimately periodic, the number of states of $\mathcal{A}$ should provide an upper bound on the possible period and preperiod of $X$. Roughly speaking, the minimal preperiod/period should not be too large compared with the size of $\mathcal{A}$. This should leave us with a finite number of candidates to test. Thanks to Proposition~\ref{pro:up}, one therefore builds a DFA for each pair of admissible preperiod/period. Equality of regular languages being decidable, we compare the language accepted by this DFA and the one accepted by $\mathcal{A}$. If an agreement is found, then $X$ is ultimately periodic, otherwise it is not. As a consequence of Proposition~\ref{bound:prep}, we only focus on the admissible periods.

For an ultimately periodic set $X\subseteq\N$, we consider the prime decomposition of its period $\pi_X$. There are two types of prime factors.
\begin{enumerate}
  \item[(T1)] Those that do not simultaneously divide all the coefficients of the recurrence relation. 
  \item[(T2)] The primes dividing all the coefficients of the recurrence relation.
\end{enumerate}
Our strategy is to bound those two types of factors separately. We depart from the strategy developed in \cite{BCFR} because we have to deal with the case of what we call a zero period discussed below.

\subsection{Prime factors of the period that do not divide all the coefficients of the recurrence relation}
\label{sec:fpc}

{ If a prime factor $p$ of the candidate period for $X$ does not divide all the coefficients of the recurrence relation, we will show that, for some integer $\mu\ge1$, the periodic part of the sequence $(U_i\bmod p^\mu)_{i\ge 0}$ contains a non-zero element. This fact will provide us with an upper bound on $p$ and its exponent in the prime decomposition of the candidate period.}

\begin{definition}
We say that an ultimately periodic sequence has a \emph{zero period} (or, zero periodic part) if it has period $1$ and the repeated element is $0$. Otherwise stated, the sequence has a tail of zeroes.
\end{definition}

\begin{remark}
Let $\mu\ge 1$. Observe that if the periodic part of $(U_i\bmod p^\mu)_{i\ge 0}$ contains a non-zero element, then the same property holds for all sequences $(U_i\bmod p^{\mu'})_{i\ge 0}$ with $\mu'\ge \mu$. 

Furthermore, assume that for infinitely many $\mu$, $(U_i\bmod p^\mu)_{i\ge 0}$ has a zero period. Then from the previous paragraph, we conclude that $(U_i\bmod p^\mu)_{i\ge 0}$ has a zero period for all $\mu\ge 1$. 
\end{remark}

\begin{example} We give a sequence where only finitely many sequences modulo $p^\mu$ have a zero period. 
    Take the sequence $U_0=1$, $U_1=4$, $U_2=8$ and $U_{i+2}=U_{i+1}+U_i$ for $i\ge 1$. Then the sequence $(U_i\bmod 2^\mu)_{i\ge 0}$ has a zero period for $\mu=1,2$ because of the particular initial conditions. But it is easily checked that it has a non-zero period for all $\mu\ge 3$.
\end{example}

The next result is a special instance of \cite[Thm.~32]{BCFR} and its proof turns out to be much simpler. 

\begin{theorem}\label{the:jason}
    Let $p$ be a prime. The sequence $(U_i\bmod p^\mu)_{i\ge 0}$ has a zero period for all $\mu\ge 1$ if and only if all the coefficients $a_0,\ldots,a_{k-1}$ of the linear relation~\eqref{eq:linrec} are divisible by $p$.
\end{theorem}

\begin{proof}
  {
    Let $N$ be given in Definition~\ref{def:linsys}. It is clear that if $a_0,\ldots,a_{k-1}$ are divisible by $p$, then for any choice of initial conditions $U_0,\ldots,U_{N+k-1}$, the elements $U_{N+k},\ldots,U_{N+2k-1}$ are divisible by $p$, hence the elements $U_{N+2k},\ldots,U_{N+3k-1}$ are divisible by $p^2$, and so on and so forth. Otherwise stated, for all $\mu\ge 1$ and all $i\ge N+\mu k$, $U_i$ is divisible by $p^{\mu}$.}

We turn to the converse. 
Since the sequence $(U_i)_{i\ge 0}$ {ultimately} satisfies a linear recurrence relation, the power series $$\mathsf{U}(x):=\sum_{i\ge 0}U_i\, x^i$$ is rational. By assumption, $(U_i\bmod p^\mu)_{i\ge 0}$ has a zero period for all $\mu\ge 1$. Otherwise stated, with the $p$-adic absolute value notation, $|U_i|_p\le p^{-\mu}$ for large enough $i$, i.e.\  $|U_i|_p\to 0$ as $i\to+\infty$. Recall that a series $\sum_{i\ge 0}\gamma_i$ converges in $\mathbb{Q}_p$ if and only if $\lim_{i\to+\infty}|\gamma_i|_p=0$. Hence the series $\mathsf{U}(x)$ converges in $\mathbb{Q}_p$ in the closed unit disc. Therefore, the poles $\rho_1,\ldots,\rho_r\in\mathbb{C}_p$ of $\mathsf{U}(x)$ must satisfy $|\rho_j|_p>1$ for $1\le j\le r$.

Let $P(x)=1-a_{k-1}x-\ldots -a_0x^k$ be the reciprocal polynomial of the linear recurrence relation~\eqref{eq:linrec}. By minimality of the order $k$ of the recurrence, the roots of $P$ are precisely the poles of $\mathsf{U}(x)$ with the same multiplicities. If we factor $$P(x)=(1-\delta_1x)\cdots (1-\delta_kx)$$ each of the $\delta_j$ is one of the $\frac{1}{\rho_1},\ldots,\frac{1}{\rho_r}$. For $n>0$, the coefficient of $x^n$ in $P(x)$ is an integer equal to a sum of product of elements of $p$-adic absolute value less than $1$. Since $|a+b|_p\le \max\{|a|_p,|b|_p\}$ and {$|ab|_p=|a|_p|b|_p$}, this coefficient is an integer with a $p$-adic absolute value less than $1$, i.e.\  a multiple of $p$.
\end{proof}

Thanks to Theorem~\ref{the:jason}, if $p$ is a prime not dividing all the coefficients of the recurrence relation~\eqref{eq:linrec} then there exists an integer $\lambda\ge 1$ such that the periodic part of $(U_i\bmod{p^\lambda})_{i\ge 0}$ contains a non-zero element.

\begin{prop}\label{pro:cas2a}
Assume (H1), (H2) and (H3). Let $p$ be a prime not dividing all the coefficients of the recurrence relation~\eqref{eq:linrec} and let $\lambda\ge 1$ be an integer such that the periodic part of $(U_i\bmod{p^\lambda})_{i\ge 0}$ contains a non-zero element. If $X\subseteq\mathbb{N}$ is an ultimately periodic $U$-recognizable set with period $\pi_X=p^\mu\cdot r$ where $\mu\ge\lambda$ and $r$ is not divisible by $p$, then the minimal automaton of $\rep_U(X)$ has at least $p^{\mu-\lambda+1}$ states.
\end{prop}

\begin{proof} 
We will make use of the following observation. Let $n\ge 1$. In the additive group $\left(\mathbb{Z}/p^n\mathbb{Z},+\right)$, an element $a$ has order $p^s$ with $0< s\le n$ if and only if $a=p^{n-s}\cdot m$ where $m$ is not divisible by $p$.

By assumption the periodic part of $(U_i\bmod p^\lambda)_{i\ge 0}$ contains a non-zero element $R$ of order $\ord_{p^\lambda}(R)=p^{\theta}$ for some $\theta$ such that $0<\theta\le\lambda$.
Using the above observation twice, $R=p^{\lambda-\theta}\cdot m$ for some $m$ coprime with $p$, the order of $R$ modulo~$p^\mu$ is $s:=\ord_{p^\mu}(R)=p^{\mu-\lambda+\theta}$. 

 Let us define $s$ integers $k_1,\ldots,k_s\ge 0$ and thus $s$ words $w_1,\ldots,w_s\in\{0,1\}^*$ of the following form $$w_j:=10^{k_j}10^{k_{j-1}} \cdots 10^{k_1}0^{\size{\rep_U(\pi_X)}}.$$ 
Thanks to Lemma~\ref{lem:h3}, we may impose the following conditions.
\begin{itemize}
  \item First, $k_1$ is taken large enough to ensure that $\val_U(w_1)$ is larger than the preperiod of $X$. 
  
  \item Second, $k_1, \ldots,k_s$ are taken large enough to ensure that $w_j\in\rep_U(\mathbb{N})$ for all $j$. Simply choose $k_j\ge Z$ for all $j$.

  \item Third, we can choose $k_1,\ldots,k_s$ so that the $1$'s occur at indices $m$ such that $U_m\equiv R \pmod{p^\mu}$.
\end{itemize}
Observe that $\val_U(w_j)\equiv j\cdot R\pmod{p^\mu}$. Since $p^\mu$ divides $\pi_X$, the words $w_1,\ldots,w_s$ have pairwise distinct values modulo~$\pi_X$. 

Let $i,j\in\{1,\ldots,s\}$ such that $i\neq j$. By Lemma~\ref{lem:per}, we can assume that there exists $r_{i,j}<\pi_X$ such that $\val_U(w_i)+r_{i,j}\in X$ and $\val_U(w_j)+r_{i,j}\not\in X$ (the symmetric situation is handled similarly). In particular, $\size{\rep_U(r_{i,j})}\le \size{\rep_U(\pi_X)}$. Consider the two words 
$$10^{k_i}10^{k_{i-1}} \cdots 10^{k_1}x_{i,j}
\quad\text{ and }\quad
10^{k_j}10^{k_{j-1}} \cdots 10^{k_1}x_{i,j}$$
where $$x_{i,j}=0^{|\rep_U(\pi_X)|-|\rep_U(r_{i,j})|}\rep_U(r_{i,j}).$$ The first word belongs to $\rep_U(X)$ and the second does not. Consequently, the number of states of the minimal automaton of $\rep_U(X)$ is at least $s=p^{\mu-\lambda+\theta}$. The conclusion follows since $\theta\ge 1$.
\end{proof}

{
From the above proposition, we immediately get the following.

\begin{corollary}\label{cor:new}
  Assume (H1), (H2) and (H3). Let $p>\max\{|a_0|,U_N\}$ be a prime. If $X\subseteq\mathbb{N}$ is an ultimately periodic $U$-recognizable set with period $\pi_X=p^\mu\cdot r$ where $\mu\ge 1$ and $r$ is not divisible by $p$, then the minimal automaton of $\rep_U(X)$ has at least $p^{\mu}$ states.
\end{corollary}

\begin{proof}
The sequence $(U_i)_{i\ge 0}$ satisfies the recurrence relation~\eqref{eq:linrec} for all $i\ge N$. Since $p>|a_0|$, $p$ does not divide $a_0$ and $(U_{N+i}\bmod p)_{i\ge 0}$ is purely periodic. By assumption $p>U_N$, hence the first element of the periodic part equals $U_N$ and is non-zero modulo~$p$. We conclude that the non-zero element $U_N$ occurs infinitely often in the sequence $(U_i\bmod{p})_{i\ge 0}$. Hence we may apply Proposition~\ref{pro:cas2a} with $\lambda=1$.
  \end{proof}
}

\subsection{Prime factors of the period that divide all the coefficients of the recurrence relation}

We can factor the period $\pi_X$ as
\begin{equation}
    \label{eq:fQ}
    \pi_X=Q_X\cdot p_1^{\mu_{X,1}}\cdots p_t^{\mu_{X,t}}
\end{equation}
where every $p_j$ divides all the coefficients of the recurrence relation~\eqref{eq:linrec} and, for every prime factor $q$ of $Q_X$, at least one of the coefficients of the recurrence relation~\eqref{eq:linrec} is not divisible by $q$.
Otherwise stated, the factor $Q_X$ collects the prime factor of type (T1). { Note that the primes $p_j$ depend only on the numeration system $U$ (i.e.\ the coefficients of the recurrence) and their exponents depend on $\pi_X$ thus, on $X$.}

\begin{remark}
There is a finite number of primes dividing all the coefficients of the recurrence relation. Thus, we only have to obtain an upper bound on the corresponding exponents $\mu_{X,1},\ldots,\mu_{X,t}$ that may appear in~\eqref{eq:fQ}. 
\end{remark}

\begin{definition}
Let $j\in\{1,\ldots,t\}$ and $\mu\ge 1$. From Theorem~\ref{the:jason}, the sequence $(U_i\bmod p_j^{\mu})_{i\ge 0}$ has a zero period. We let $\mathsf{f}_{p_j}(\mu)$ denote the length of the preperiod, i.e.\  $U_{\mathsf{f}_{p_j}(\mu)-1}\not\equiv 0\pmod{ p_j^{\mu}}$ and $U_{i}\equiv 0\pmod{p_j^{\mu}}$ for all $i\ge \mathsf{f}_{p_j}(\mu)$. 
\end{definition}

\begin{example}
    Let us consider the numeration system from Example~\ref{exa:ppp}. The sequence $(U_i\bmod 2)_{i\ge 0}$ is $1,1,1,1,0^\omega$. Hence $\mathsf{f}_2(1)=4$.  The sequence $(U_i\bmod 4)_{i\ge 0}$ is $1, 3, 1, 3, 2, 0, 2, 2, 0^\omega$. Hence $\mathsf{f}_2(2)=8$. Continuing this way, we have $\mathsf{f}_2(3)=12$ and $\mathsf{f}_2(4)=16$. 
\end{example}

Note that $\mathsf{f}_{p_j}$ is non-decreasing:
$\mathsf{f}_{p_j}(\mu+1)\ge \mathsf{f}_{p_j}(\mu)$ {
  and
  \begin{equation}\label{eq:limfpj}
  \lim_{\mu\to+\infty} \mathsf{f}_{p_j}(\mu)=+\infty.    
  \end{equation}
}

{
\begin{lemma}{\cite[Lemma~6]{Honkala}}\label{lem:hon}
Let $X$ be an ultimately periodic set with period \eqref{eq:fQ}. There exists $r \in \{ 0,\ldots, Q_X {-}1 \}$ such that $X \cap (Q_X \N + r)$ is ultimately periodic of period $Q_X \cdot p_1^{\nu_1} \cdots p_t^{\nu_t}$ with $$\max\limits_{1\le j\le t} \mu_{X,j} = \max\limits_{1\le j\le t} \nu_j.$$
\end{lemma}

}

\begin{definition}\label{de:Qpinu}
  {
    The quantity $r$ in the previous lemma is not necessarily unique. To avoid ambiguity, we always consider the smallest possible such $r$ denoted by $r_X$ and the associated exponents $\nu_{X,1}, \ldots , \nu_{X,t}$.  We therefore let $\rho_X$ denote the corresponding quantity $Q_X \cdot p_1^{\nu_{X,1}} \cdots p_t^{\nu_{X,t}}$.}
  
We let ${M_{\bm{\mu},X}}$ denote the maximum of the values $\mathsf{f}_{p_j}(\nu_{X,j})$ for $j\in\{1,\ldots,t\}$:
$${M_{\bm{\mu},X}}=\max_{1\le j\le t}\mathsf{f}_{p_j}(\nu_{X,j}).$$ 
Thus, ${M_{\bm{\mu},X}}$ is the least index such that for all $i\ge M_{\bm{\mu},X}$ and all $j\in\{1,\ldots,t\}$, $U_i\equiv 0 \pmod{p_j^{\nu_{X,j}}}$. By the Chinese remainder theorem, ${M_{\bm{\mu},X}}$ is also the least index such that for all $i\ge M_{\bm{\mu},X}$, $$U_i\equiv 0 \pmod{\frac{\rho_X}{Q_X}}.$$
{
  The reader may notice that $M_{\bm{\mu},X}$ only depends on the exponents $\bm{\mu}=(\mu_{X,1},\ldots,\mu_{X,t})$ occurring in \eqref{eq:fQ}.}
\end{definition}

{
From Lemma~\ref{lem:hon} and \eqref{eq:limfpj}, for each $j\in\{1,\ldots,t\}$, $\lim_{\mu_{X,j}\to+\infty} M_{\bm{\mu},X}=+\infty$.
}


\begin{example}
    Let us consider the numeration system from Example~\ref{exa:toy}. Here we have two prime factors $2$ and $3$ to take into account. Computations show that  $\mathsf{f}_2(1)=3$, $\mathsf{f}_2(2)=5$, $\mathsf{f}_2(3)=7$ and $\mathsf{f}_3(1)=3$, $\mathsf{f}_3(2)=6$, $\mathsf{f}_3(3)=9$. Assume that we are interested in a period $\rho_X/Q_X=72=2^3\cdot 3^2$. With the above definition, $M_{\bm{\mu},X}=\max\{\mathsf{f}_2(3),\mathsf{f}_3(2)\}=7$. One can check that $(U_i\bmod 72)_{i\ge 0}$ is $1, 13, 19, 30, 54, 48, 36, 0^\omega$. 
\end{example}

We introduce a quantity $\gamma_{Q_X}$ which only depends on the numeration system $U$ and the number $Q_X$ defined in~\eqref{eq:fQ}. Since we are only interested in decidability issues, there is no need to find a sharp estimate on this quantity.

\begin{definition}\label{def:gamma}
{Let $Q\ge 1$ be an integer.} Under (H1), for each $r\in\{0,\ldots,Q-1\}$, 
a DFA accepting the language $\rep_U(Q\, \mathbb{N}+r)$
can be effectively built (see Proposition~\ref{pro:up} or the construction in \cite[Prop.~3.1.9]{cant}). We let $\gamma_Q$ denote the maximum of the numbers of states of these DFAs for $r\in\{0,\ldots,Q-1\}$. 
\end{definition}

The crucial point in the next statement is that the most significant digit $1$ occurs for $U_{M_{\bm{\mu},X}-1}$ in a specific word. The proof makes use of the same kind of arguments built for definite languages as in \cite[Lemma~2.1]{LRRV}.

\begin{theorem}\label{the:main}
  {
Assume (H1), (H2) and (H3). Let $X\subseteq\mathbb{N}$ be an ultimately periodic $U$-recognizable set with period $\pi_X$ factored as in \eqref{eq:fQ}. 
Assume that ${M_{\bm{\mu},X}}-1 -\vert\rep_U(\rho_X-1)\vert \ge Z$, where $Z$ is the constant given in Definition~\ref{def:Z} and $M_{\bm{\mu},X}$ and $\rho_X$ are given in Definition~\ref{de:Qpinu}. Also assume that $M_{\bm{\mu},X}$ is greater than the preperiod of $(U_i \bmod{Q_X})_{i \in \N}$. Then the minimal automaton of $0^*\rep_U(X)$ has at least $\frac{\vert \rep_U ( \rho_X - 1) \vert + 1}{\gamma_{Q_X}}$ states.
}
\end{theorem}

This result will provide us with an upper bound on $\mu_{X,1},\ldots,\mu_{X,t}$ (details are given in Section~\ref{ss52}). If $\max_j \mu_{X,j}=\max_j \nu_{X,j}\to\infty$, then  $\rho_X\to\infty$ and since $Q_X$ has been bounded in the first part of this paper, the number of states of the minimal automaton of $\rep_U(X)$ should increase.

\begin{proof}
  {
    We may apply Lemma~\ref{lem:h3}: if $w$ is a greedy $U$-representation, then, for all $z\geq Z$, $10^zw$ also belongs to $\rep_U(\mathbb{N})$. Let $r_X$ be the quantity given in Definition~\ref{de:Qpinu}. The set $X\cap (Q_X\N+r_X)$ has period $\rho_X$. Let $\mathcal{B}_X$ be the minimal automaton of $0^*\rep_U(X\cap (Q_X\N+r_X))$. We will provide a lower bound on the number of states of this automaton.  Let $g$ be large enough so that
\begin{itemize}
\item $g \ge Z$
\item $U_{M_{\bm{\mu},X}+g}$ is larger than the preperiod of $X\cap (Q_X\N+r_X)$
\item $g+1$ is a multiple of the period of $(U_i \bmod{Q_X})_{i\in\N}$.
\end{itemize}
Consider
\begin{align*}
n_1 &= \val_U((10^g)^{Q_X}10^{M_{\bm{\mu},X}-1}) = \sum_{i=0}^{Q_X}U_{M_{\bm{\mu},X} -1+ i(g+1)}\\
n_2 &= \val_U ( 10^{M_{\bm{\mu},X}+g} ) = U_{M_{\bm{\mu},X}+g}.
\end{align*}
Observe that $n_1$ and $n_2$ are both congruent to $U_{M_{\bm{\mu},X}-1}$ modulo $Q_X$ (we make use of the assumption that $M_{\bm{\mu},X}$ is greater than the preperiod of $(U_i \bmod{Q_X})_{i \in \N}$). However, by definition of ${M_{\bm{\mu},X}}$, $$n_1 \bmod\frac{\rho_X}{Q_X} = U_{M_{\bm{\mu},X}-1} \bmod \frac{\rho_X}{Q_X}\neq 0$$but $n_2$ is congruent to~$0$  modulo $\frac{\rho_X}{Q_X}$. Consequently, $n_1$ and $n_2$ are not congruent modulo $\rho_X$. By Lemma~\ref{lem:per} applied to the set $X\cap(Q_X\N+r_X)$, we may suppose that there exists $s<\rho_X$ such that 
$$n_1+s\in X\cap(Q_X\N+r_X)   \quad \text{and}\quad n_2+s\not\in X\cap(Q_X\N+r_X)$$ (the symmetrical situation can be treated in the same way). 
By assumption, $M_{\bm{\mu},X}-1-\vert\rep_U(s)\vert\ge M_{\bm{\mu},X}-1-\vert\rep_U(\rho_X{-}1)\vert\ge Z$.} Thanks to Lemma~\ref{lem:h3}, both words
    {
$$u=(10^g)^{Q_X} 1 0^{{M_{\bm{\mu},X}}{-}1{-}\vert\rep_U(s)\vert}\rep_U(s)$$ and
$$v=10^{g} 0 0^{{M_{\bm{\mu},X}}{-}1{-}\vert\rep_U(s)\vert}\rep_U(s)$$
}are greedy $U$-representations. For all $\ell\ge 0$, define an equivalence relation $E_\ell$ on the set of states of~$\mathcal{B}_X$: 
$$E_\ell(q,q')\Leftrightarrow (\forall x\in A_U^*)\bigl[|x|\ge \ell\Rightarrow (\delta(q,x)\in \mathcal{F}\Leftrightarrow \delta(q',x)\in \mathcal{F})\bigr]$$ where $\delta$ (resp.\ $\mathcal{F}$) is the transition function (resp.\ the set of final states) of $\mathcal{B}_X$. Let us denote the number of equivalence classes of $E_\ell$ by $P_\ell$. Clearly, $E_\ell(q,q')$ implies $E_{\ell+1}(q,q')$, and thus $P_\ell\ge P_{\ell+1}$. {One can already observe that $P_0$ is the number of states of $\mathcal{B}_X$.}

Let $i\in\{0,\ldots,\vert\rep_U(\rho_X{-}1)\vert\}$. By assumption, $\vert\rep_U(\rho_X{-}1)\vert<M_{\bm{\mu},X}$. Since $u$ and $v$ have the same suffix of length $M_{\bm{\mu},X}-1$, we can factorize these words as
$$u=u_i w_i
\quad\text{and}\quad
v=v_i w_i$$
where $|w_i|=i$. Let $q_0$ be the initial state of $\mathcal{B}_X$. By construction, $\delta(q_0,u_i w_i)\in\mathcal{F}$ whereas $\delta(q_0,v_i w_i)\notin\mathcal{F}$, hence the states $\delta(q_0,u_i)$ and $\delta(q_0,v_i)$ are not in relation with respect to $E_i$. Let us show that, for all $j> i$, they satisfy $E_{j}$. It is enough to show that
\begin{equation}
\label{eq:Ei+1}
E_{i+1}(\delta(q_0,u_i),\delta(q_0,v_i)).
\end{equation}

Figures~\ref{fig:ti} and~\ref{fig:ti2} can help the reader. Let $x$ be such that $|x|=i+t$, with $t\ge 1$. Let $p$ be the prefix of $\rep_U(s)$ of length $\size{\rep_U(s)}-i$, this prefix~$p$ being empty whenever this difference is negative. 
If we replace $w_i$ by $x$ in $u$ and $v$, we get 
{
$$u_ix=(10^g)^{Q_X}10^{M_{\bm{\mu},X}{-}1{-}\vert px\vert +t}px
\quad\text{and}\quad
v_ix=10^g00^{M_{\bm{\mu},X}{-}1{-}\vert px \vert+t}px.$$
}
Then $$\val_U(u_i x)-\val_U(v_i x)=U_{M_{\bm{\mu},X}+t{-}1} + \sum_{i=2}^{Q_X} U_{M_{\bm{\mu},X}-1 + i(g+1) + t}.$$
Since by assumption, $M_{\bm{\mu},X}$ is larger than the preperiod of $(U_i \bmod{Q_X})_{i\in\N}$, 
this quantity is congruent to $0$ modulo $Q_X$ and by definition of ${M_{\bm{\mu},X}}$, it is also congruent to $0$ modulo ${\frac{\rho_X}{Q_X}}$. Hence, $\val_U(u_i x)$ and $\val_U(v_i x)$ belong to the periodic part of $X\cap (Q_X\N+r_X)$ and they differ by a multiple of the period $\rho_X$. Therefore, $\val_U(u_i x)$ belongs to $X\cap (Q_X\N+r_X)$ if and only if $\val_U(v_i x)$ also does. 
\begin{figure}[h!t]
    \centering
\begin{tikzpicture}
\draw [dashed] (3.05,.6) -- (3.05,-2.4);
\draw [dashed] (.7,2.2) -- (.7,1.2);
\draw [dashed] (.7,.5) -- (.7,-0.5);
\draw (-1,0) rectangle (3,.6);
\node [above left] at (-1,0) {$u$:};
\draw [<->] (3.1,.8) -- (5.1,.8);
\node [above] at (4.1,.8) {$i$};
\draw [<->] (2,1.5) -- (5.1,1.5);
\node [above] at (3.55,1.5) {$\le\vert\rep_U(\rho_X{-}1)\vert$};
\draw [<->] (1,2.2) -- (5.1,2.2);
\node [above] at (3.55,2.2) {$M_{\bm{\mu},X}-1$};
\node [left] at (.95,.9) {$1$};
\node [left] at (.95,-.7) {$0$};
\draw (3.1,0) rectangle (5.1,.6);
\node [above] at (1,0) {$u_i$}; 
\node [above] at (4.1,0) {$w_i$}; 
\draw (2,-.8) rectangle (5.1,-.2);
\node [above left] at (-1,-.8) {$\rep_U(s)$:};
\node [above] at (2.5,-.8) {$p$}; 
\draw (-1,-1.6) rectangle (3,-1);
\draw (3.1,-1.6) rectangle (5.1,-1);
\node [above] at (1,-1.6) {$v_i$};
\node [above left] at (-1,-1.6) {$v$:}; 
\node [above] at (4.1,-1.6) {$w_i$};
\draw (3.1,-2.4) rectangle (6.5,-1.8);
\draw [<->] (5.2,-1.6) -- (6.5,-1.6);
\node [above] at (5.85,-1.6) {$t$};
\node [above] at (4.8,-2.4) {$x$};
\end{tikzpicture}    
    \caption{The different words (case where $i\le \size{\rep_U(s)}$).}
    \label{fig:ti}
\end{figure}

\begin{figure}[h!t]
    \centering
\begin{tikzpicture}
\draw [dashed] (3.05,.6) -- (3.05,-2.4);
\draw [dashed] (.7,2.2) -- (.7,1.2);
\draw [dashed] (.7,.5) -- (.7,-0.5);
\draw (-1,0) rectangle (3,.6);
\node [above left] at (-1,0) {$u$:};
\draw [<->] (3.1,.8) -- (5.1,.8);
\node [above] at (4.1,.8) {$i$};
\draw [<->] (2,1.5) -- (5.1,1.5);
\node [above] at (3.55,1.5) {$\le\vert\rep_U(\rho_X{-}1)\vert$};
\draw [<->] (1,2.2) -- (5.1,2.2);
\node [above] at (3.55,2.2) {$M_{\bm{\mu},X}-1$};
\node [left] at (.95,.9) {$1$};
\node [left] at (.95,-.7) {$0$};
\draw (3.1,0) rectangle (5.1,.6);
\node [above] at (1,0) {$u_i$}; 
\node [above] at (4.1,0) {$w_i$}; 
\draw (3.5,-.8) rectangle (5.1,-.2);
\node [above left] at (-1,-.8) {$\rep_U(s)$:};
\draw (-1,-1.6) rectangle (3,-1);
\draw (3.1,-1.6) rectangle (5.1,-1);
\node [above] at (1,-1.6) {$v_i$};
\node [above left] at (-1,-1.6) {$v$:}; 
\node [above] at (4.1,-1.6) {$w_i$};
\draw (3.1,-2.4) rectangle (6.5,-1.8);
\draw [<->] (5.2,-1.6) -- (6.5,-1.6);
\node [above] at (5.85,-1.6) {$t$};
\node [above] at (4.8,-2.4) {$x$};
\end{tikzpicture}    
    \caption{The different words (case where $i> \size{\rep_U(s)}$).}
    \label{fig:ti2}
\end{figure}
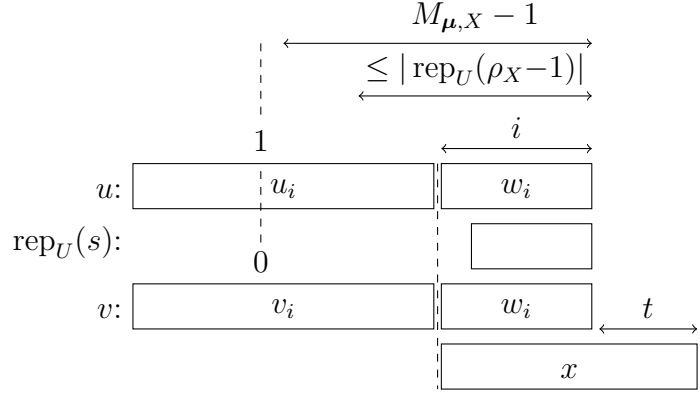
{
In order to obtain~\eqref{eq:Ei+1}, it remains to show that either both $u_ix$ and $v_ix$ are valid greedy $U$-representations or both are not. If the word $px$ is not a greedy $U$-representation then neither $u_ix$ nor $v_ix$ can be valid. Assume now that $px$ is a greedy $U$-representation. 
Note that in both situations described in Figures~\ref{fig:ti} and~\ref{fig:ti2}, $\vert px\vert \le \vert\rep_U(\rho_X{-}1)\vert+t$. Thanks to the assumption, we obtain $M_{\bm{\mu},X}-1-\vert px \vert +t\ge M_{\bm{\mu},X}-1-\vert\rep_U(\rho_X{-}1)\vert \ge Z$. The greediness of $px$ and Lemma~\ref{lem:h3} imply that 
$10^{M_{\bm{\mu},X}{-}1{-}\vert px\vert +t}px$
is a greedy $U$-representation. Since $g\ge Z$, $u_ix$ is also a greedy $U$-representation and the same observation trivially holds for $v_ix$.

We conclude that 
$$P_0 > P_1 > \cdots > P_{\vert\rep_U(\rho_X{-}1)\vert}\ge 1.$$
Since $P_0$ is the number of states of $\mathcal{B}_X$, the automaton $\mathcal{B}_X$ has at least $\vert\rep_U(\rho_X{-}1)\vert+1$ states.

Finally, let $\mathcal{A}_X$ and $\mathcal{A}_r$ be the minimal automata of $0^*\rep_U(X)$ and $0^*\rep_U(Q_X\mathbb{N}+r_X)$ respectively. The number of states of $\mathcal{A}_r$ is bounded by $\gamma_{Q_X}$. The DFA $\mathcal{B}_{X}$ is a quotient of the product automaton $\mathcal{A}_X\times \mathcal{A}_r$, hence the number of states of $\mathcal{B}_{X}$ is at most the number of states of $\mathcal{A}_X$ times $\gamma_{Q_X}$. We thus obtain that the number of states of $\mathcal{A}_X$ is at least $\frac{\vert\rep_U(\rho_X{-}1)\vert+1}{\gamma_{Q_X}}$.
}
\end{proof}


\section{Cases we can deal with}\label{sec5}

\subsection{The gcd of the coefficients of the recurrence relation is~$1$.}

In this case, for any ultimately periodic set $X$, the 
factorization of the period $\pi_X$ given in \eqref{eq:fQ} has the special form $\pi_X=Q_X$ and the addressed decision problem turns out to be decidable.

\begin{theorem}\label{thm:main1}
    Let $U$ be a linear numeration system satisfying (H1), (H2) and (H3), and such that the  gcd of the coefficients of the recurrence relation~\eqref{eq:linrec} is $1$. Given a DFA  accepting a language contained in the numeration language $\rep_U(\mathbb{N})$, it is decidable whether this DFA recognizes an ultimately periodic set.
\end{theorem}

\begin{proof} Let $\mathcal{A}$ be a DFA accepting a language contained in the numeration language. Let $X$ be the set of integers recognized by $\mathcal{A}$.
  
  Assume that $X$ is an ultimately periodic set with period $\pi_X$. 
Let $p$ be a prime that divides $\pi_X$. {Either $p\le \max\{|a_0|,U_N\}$ or $p> \max\{|a_0|,U_N\}$.

In the former case, there is only a finite number of such primes. By assumption, $p$ does not divide all the coefficients of the recurrence relation. Then thanks to Theorem~\ref{the:jason}, there exists $\lambda\ge 1$ such that the periodic part of the sequence $(U_i\bmod{p^\lambda})_{i\ge 0}$ contains a non-zero element. By an exhaustive search, one can determine the value of the least such $\lambda$: one finds the period of a sequence $(U_i\bmod{p^\lambda})_{i\ge 0}$ as soon as two $k$-tuples $(U_i \bmod{p^\lambda},\ldots,U_{i+k-1} \bmod{p^\lambda})$ are identical (where $k$ is the order of the recurrence). 
We then apply Proposition~\ref{pro:cas2a}. For any $\mu\ge 1$, if $p^\mu$ divides $\pi_X$ then either $\mu<\lambda$ or $p^{\mu-\lambda+1}$ is bounded by the number $S$ of states of $\mathcal{A}$. So we have bounded the exponent $\mu$ of those primes that may occur in $\pi_X$  by $\max\{\lambda,\log_p(S)+\lambda-1\}$.

In the latter case, thanks to Corollary~\ref{cor:new}, for any $\mu\ge 1$, if $p^\mu$ divides $\pi_X$ then $p^\mu$ is bounded by the number of states of $\mathcal{A}$. }

The previous discussion provides us with an upper bound on $\pi_X$, i.e.\ on the admissible periods for $X$. Then from Proposition~\ref{bound:prep}, associated with each admissible period, there is a computable bound for the corresponding admissible preperiods for $X$. We conclude that there is a finite number of pairs of candidates for the preperiod and period of $X$.  Similar to Honkala's scheme, we therefore have a decision procedure by enumerating a finite number of candidates. For each pair $(a,b)$ of possible preperiods and periods, there are $2^a2^b$ corresponding ultimately periodic sets $X$. For each such candidate $X$, we build a DFA accepting $\rep_U(X)$ and compare it with $\mathcal{A}$. We can conclude since equality of regular languages is decidable. 
\end{proof}

There exist recurrence relations {satisfying the assumptions of the above theorem} but that were not handled in \cite{BCFR}. Take \cite[Example~35]{BCFR} 
$$U_{i+5} = 6 U_{i +4} + 3U_{i+3} - U_{i+2} + 6U_{i+1} +  3U_{i},\ \forall i\ge 0.$$ 
For this recurrence relation, $\mathcal{N}_U(3^i)\not\to\infty$. The characteristic polynomial has the dominant root $3+2\sqrt{3}$ and it also has three roots of modulus $1$. Therefore, no decision procedure was known. 
But thanks to Theorem~\ref{thm:main1}, we can handle such new cases under our mild assumptions (H1), (H2) and (H3). Indeed, by applying Bertrand's theorem with the initial conditions $1,7,45,291,1881$, the numeration language $0^*\rep_U(\mathbb{N})$ is the set of words over $\{0,1,\ldots,6\}$ avoiding the factors $63,64,65,66$, hence (H1) holds. Moreover, it is easily checked that for all $i\ge 0$, $U_{i+1}-U_i\ge 5 U_i$. Therefore, the system $U$ also satisfies (H2) and (H3).

\subsection{The gcd of the coefficients of the recurrence relation is larger than 1.}\label{ss52} 
If $X$ is an ultimately periodic set with period $\pi_X=Q_X\cdot p_1^{\mu_{X,1}}\cdots p_t^{\mu_{X,t}}$ with $t\ge 1$ as in~\eqref{eq:fQ}, then the quantity $M_{\bm{\mu},X}$ is well defined. Theorem~\ref{the:main} has a major assumption. 
{
 The quantity\index{$n_X$}
\[
n_X=M_{\bm{\mu},X}- 1 - \left\vert\rep_U\left(\rho_X-1\right)\right\vert
\]
should be larger than some positive constant $Z$, which only depends on the numeration system $U$.
}
\begin{theorem}
\label{thm:main2} 
Let $U$ be a linear numeration system satisfying (H1), (H2) and (H3), and such that the gcd of the coefficients of the recurrence relation~\eqref{eq:linrec} is larger than 1.
{
  Let $Z$ be the constant given in Definition~\ref{def:Z}. }
Assume there exists a computable positive integer $D$ such that for all ultimately periodic sets $X$ of period $\pi_X=Q_X\cdot p_1^{\mu_{X,1}}\cdots p_t^{\mu_{X,t}}$ as in~\eqref{eq:fQ} with $t\ge 1$, if $\max(\mu_{X,1},\ldots,\mu_{X,t})\ge D$ then $n_X\ge Z$. Then, given a DFA accepting a language contained in the numeration language $\rep_U(\mathbb{N})$, it is decidable whether this DFA recognizes an ultimately periodic set.
\end{theorem}

\begin{proof} Let $\mathcal{A}$ be a DFA accepting a language contained in the numeration language. Let $X$ be the set of integers recognized by $\mathcal{A}$.
  
    Assume that $X$ is an ultimately periodic set with period $\pi_X=Q_X\cdot p_1^{\mu_{X,1}}\cdots p_t^{\mu_{X,t}}$ as in~\eqref{eq:fQ}. Note that there are only finitely many primes dividing all the coefficients of the recurrence relation~\eqref{eq:linrec}, hence the possible $p_1,\ldots,p_t$ belong to a finite set depending only on the numeration system $U$.
    
    Applying the same reasoning as in the proof of Theorem~\ref{thm:main1}, $Q_X$ is bounded by a constant $B$ deduced from $\mathcal{A}$. So the quantity~$\gamma_{Q_X}$ introduced in Definition~\ref{def:gamma} is also bounded. 

    {
Compute the greatest preperiod $P$ of the sequences $(U_i \bmod{b})_{i \in \N}$, for $b \in \{1,\ldots, B \}$. Then by definition of $M_{\bm{\mu},X}$, there exists a computable constant $D^\prime$ such that if $\max (\mu_{X,1}, \cdots , \mu_{X,t}) \geq D^\prime$, then $M_{\bm{\mu},X}$ is greater than $P$.

By hypothesis, there is a computable positive integer constant $D$ such that if $\max(\mu_{X,1}, \cdots,\mu_{X,t})\ge D$ then $n_X\ge Z$. Let $E = \max (D, D^{\prime})$. The number of $t$-uples $(\mu_{X,1}, \cdots,\mu_{X,t})$ in $ \{0,\ldots,E{-}1\}^t$ is finite. Hence there is a finite number of periods $\pi_X$ of the form $Q_X\cdot p_1^{\mu_{X,1}}\cdots p_t^{\mu_{X,t}}$ with $Q_X$ bounded by $B$ and $(\mu_{X,1}, \cdots,\mu_{X,t})$ in this set. We can enumerate them and proceed as in the last paragraph of the proof of Theorem~\ref{thm:main1}.

We may now assume that $\max(\mu_{X,1}, \cdots,\mu_{X,t})\ge E$. In this case, $n_X\ge Z$. Moreover, $M_{\bm{\mu},X}$ is greater than $P$. We are thus able to apply Theorem~\ref{the:main}\footnote{Considering leading zeroes or not does not change the reasoning.}: it provides a bound on $\rho_X$ and thus on the possible exponents $\mu_{X,1},\ldots,\mu_{X,t}$ depending only on $\mathcal{A}$. We conclude in the same way as in the proof of Theorem~\ref{thm:main1}.
      }
\end{proof}

In the last part of this section, we present a possible way to tackle new examples of numeration systems by applying Theorem~\ref{thm:main2}. We stress the fact that when $\pi_X$ is increasing then potentially both terms $M_{\bm{\mu},X}$ and $\size{\rep_U({\rho_X}-1)}$ are increasing. If $\beta>1$ (see Definition~\ref{def:uT}), then the growth of $\size{\rep_U({\rho_X}-1)}$ has a logarithmic bound thanks to Lemma~\ref{lem:length}, so we need insight on $\mathsf{f}_{p_j}(\mu)$ to be able to guarantee $n_X\ge Z$. { In the next few pages we therefore try to obtain conditions allowing us to apply the decision procedure of Theorem~\ref{thm:main2} and, facing non-trivial number theoretic problems, we discuss how far it is possible to go.}

The \emph{$p$-adic valuation} of an integer $n$, denoted $\nu_p(n)$, is the exponent of the highest power of $p$ dividing $n$. There is a clear link between $\nu_{p_j}$ and $\mathsf{f}_{p_j}$: for all non-negative integers $\mu$ and $N$,
$$\mathsf{f}_{p_j}(\mu)=N \iff 
(\nu_{p_j}(U_{N-1})<\mu\ \wedge\ \forall i\ge N,\, \nu_{p_j}(U_i)\ge \mu).$$

\begin{remark}
With our Example~\ref{exa:toy} and initial conditions $1,2,3$, computing the first few values of $\nu_2(U_i)$, as shown in Figure~\ref{fig: third-order p=2}, might suggest that it is bounded by a function of the form $\frac{i}{2}+c$, for some constant $c$.
\begin{figure}
	\includegraphics[width=.5\textwidth]{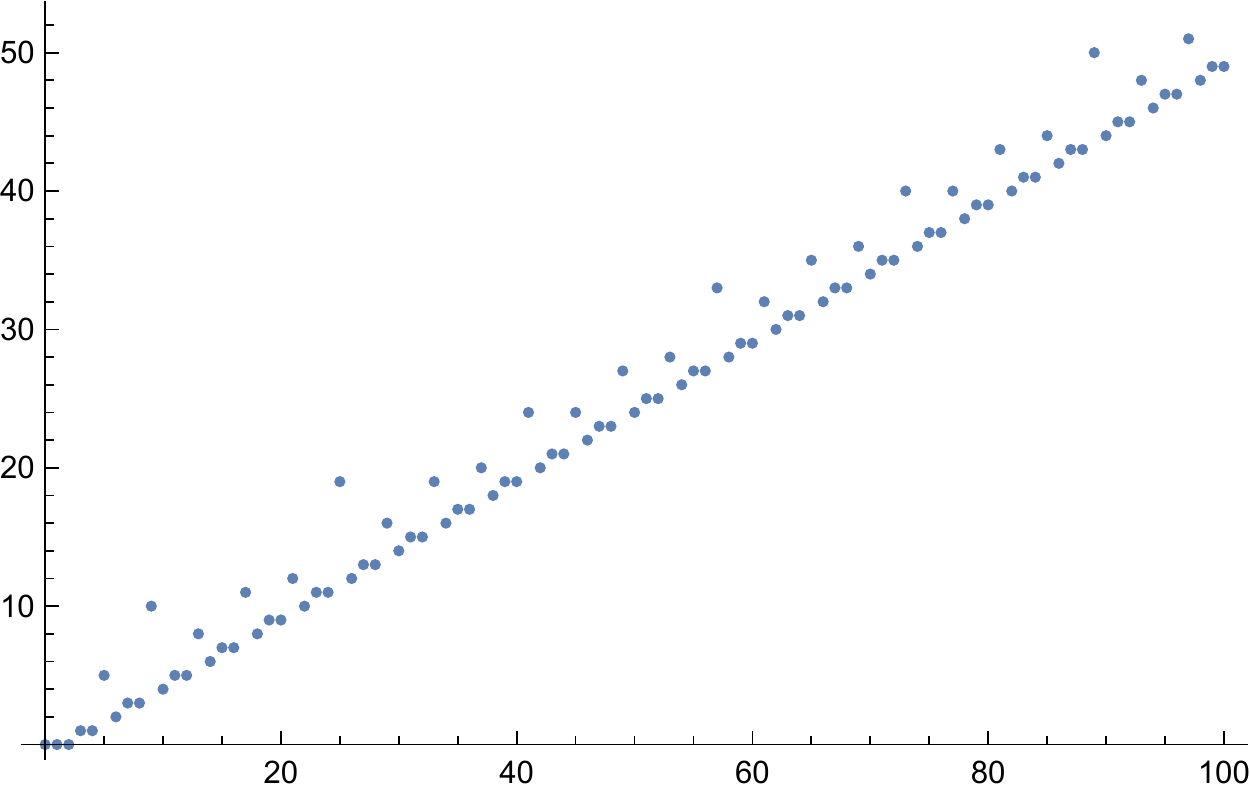}
	\caption{Plot of the $2$-adic valuation of the sequence in Example~\ref{exa:toy}.}
	\label{fig: third-order p=2}
\end{figure}
Nevertheless, computing more terms we get the following pairs $(i,\nu_2(U_i))$: $(67, 44)$, $(2115, 1070)$, $(10307, 5172)$, $(534595, 267318)$, $(2631747, 1315896)$. The constant $c$ suggested by each of these points is respectively $\frac{21}{2}$, $\frac{25}{2}$, $\frac{37}{2}$, $\frac{41}{2}$, $\frac{45}{2}$, which is increasing. This example explains the second term $g(i)$ in the function bounding $\nu_{p_j}(U_i)$ in the next statement. 
\end{remark}

In the next statement, the reader can think about logarithm function instead of a general function $g$. Indeed, for any $\epsilon>0$, for large enough~$i$, $\log(i)< \epsilon\, i$. We also keep context and notation from \eqref{eq:fQ}.

\begin{lemma}\label{lem:valp}
Let $j\in\{1,\ldots,t\}$ and let $\beta$ as in Definition~\ref{def:uT}. Assume that $\beta>1$ and that there exist $\alpha,\epsilon\in\mathbb{R}_{>0}$ and a non-decreasing function $g$ such that
$$\nu_{p_j}(U_i)<\lfloor \alpha i\rfloor +g(i)$$
and there exists $N$ such that $g(i)<\epsilon\, i$ for all $i>N$. Then, for large enough $\mu$, 
$$\mathsf{f}_{p_j}(\mu)> \frac{\mu}{\alpha+\epsilon}.$$
\end{lemma}

\begin{proof}
By definition of the $p$-adic valuation, $p_j^{\nu_{p_j}(U_i)} \mid U_i$ and $p_j^{\nu_{p_j}(U_i)+1}\nmid U_i$. Thus, 
by definition of $\mathsf{f}_{p_j}$, for all $i$,  
$$\mathsf{f}_{p_j}(\nu_{p_j}(U_i)+1)\ge i+1.$$ 
For all $\mu$, there exists $i$ such that 
$$\lfloor \alpha i\rfloor+g(i) \le \mu < \lfloor \alpha (i+1)\rfloor+g(i+1).$$
Take $\mu$ large enough so that $i\ge N$. 
Using the right-hand side inequality, $\mu< \alpha(i+1)+\epsilon(i+1)$ and we get 
$$i> \frac{\mu}{\alpha+\epsilon}-1.$$
Using the left-hand side inequality,  $\mu\ge \lfloor \alpha i\rfloor+g(i) > \nu_{p_j}(U_i)$. Since we have integers on both sides, $\mu\ge \nu_{p_j}(U_i)+1$. 
Since $\mathsf{f}_{p_j}$ is non-decreasing, for all large enough $\mu$, 
\[\mathsf{f}_{p_j}(\mu) \ge \mathsf{f}_{p_j}(\nu_{p_j}(U_i)+1)\ge i+1> \frac{\mu}{\alpha+\epsilon}. \qedhere\]
\end{proof}

We look for a lower bound for $n_X$. 
Suppose that for each $j\in\{1,\ldots,t\}$, there exists $\alpha_j,\epsilon_j,g_j$ and $N_j$ as in the above lemma. {Then
$${M_{\bm{\mu},X}}=\max_j \mathsf{f}_{p_j}(\nu_{X,j}) 
>\max_j \left(\frac{\nu_{X,j}}{\alpha_j+\epsilon_j}\right) 
\ge \frac{\max_j \nu_{X,j}}{\max_j (\alpha_j+\epsilon_j)}.$$
}Second, let $u
$ and $\beta$ as in Definition~\ref{def:uT}. By hypothesis, $\beta>1$. Applying Lemma~\ref{lem:length}, there exists a constant $K$ such that
$$\size{\rep_U(\frac{\pi_X}{Q_X}-1)}\le u \log_{\beta} \left(\prod_j p_j^{\mu_{X,j}}\right)+K.$$
The right hand side is
$$u \sum_j \mu_{X,j} \log_{\beta} (p_j) +K \le u (\max_j \mu_{X,j}) \sum_j \log_{\beta} p_j +K.$$
Recall that  $\max_j \nu_{X,j}=\max_j \mu_{X,j}$ (see Lemma~\ref{lem:hon}). Consequently, 
$$n_X 
\ge \max_j \mu_{X,j} \left( \frac{1}{\max_j (\alpha_j+\epsilon_j)} - u \sum_j \log_{\beta} p_j\right) -K-1.$$

If $\pi_X$ tends to infinity (and assuming that the corresponding factor~$Q_X$ remains bounded as explained in the proof of Theorem~\ref{thm:main2}), then $\max_j \mu_{X,j}$ must also tend to infinity. So we are able to conclude, i.e.\  $n_X$ tends to infinity and in particular, $n_X$ will {
  become larger than~$Z$} (the constant from Definition~\ref{def:Z}) whenever
\begin{equation}
    \label{eq:test}
    \frac{1}{\max_j (\alpha_j+\epsilon_j)} > u \sum_j \log_{\beta} p_j.
\end{equation}

Actually, we don't need $n_X$ tending to infinity, we have the weaker requirement $n_X\ge Z$. The constant $D$ from Theorem~\ref{thm:main2} can be obtained as follows. To ensure that $n_X\ge Z$, it is enough to have 
\begin{equation}
    \label{eq:test2}
\max_j \mu_{X,j} \ge \frac{Z+K+1}{\frac{1}{\max_j (\alpha_j+\epsilon_j)} - u \sum_j \log_{\beta} p_j}
\end{equation}
and the right hand side only depends on the numeration system $U$.

As a conclusion, we simply define the constant $D$ as the right hand side in \eqref{eq:test2} and, under the assumption of Lemma~\ref{lem:valp} about the behavior of the $p_j$-adic valuations of $(U_i)_{i\ge 0}$, the decision procedure of Theorem~\ref{thm:main2} may thus be applied. From a practical point of view, even though $n_X$ tending to infinity is not required, trying to make a conjecture on \eqref{eq:test} is relatively easy as seen in the following remark. This is not a formal proof, simply rough computations suggesting what could be the value of $\alpha$ in Lemma~\ref{lem:valp}.

\begin{remark}
    One can first make some computational experiments. Take the numeration system of Example~\ref{exa:ppp}. If we compute $\nu_2(U_i)$, the values for $41\le i\le 60 $ are given by 
$$ 10, 10, 10, 11, 12, 11, 11, 12, 12, 12, 12, 13, 16, 13, 13, 14, 14, 14, 14, 15.$$
This sequence is plotted in Figure~\ref{fig: fourth-order p=2}.
\begin{figure}
	\includegraphics[width=.5\textwidth]{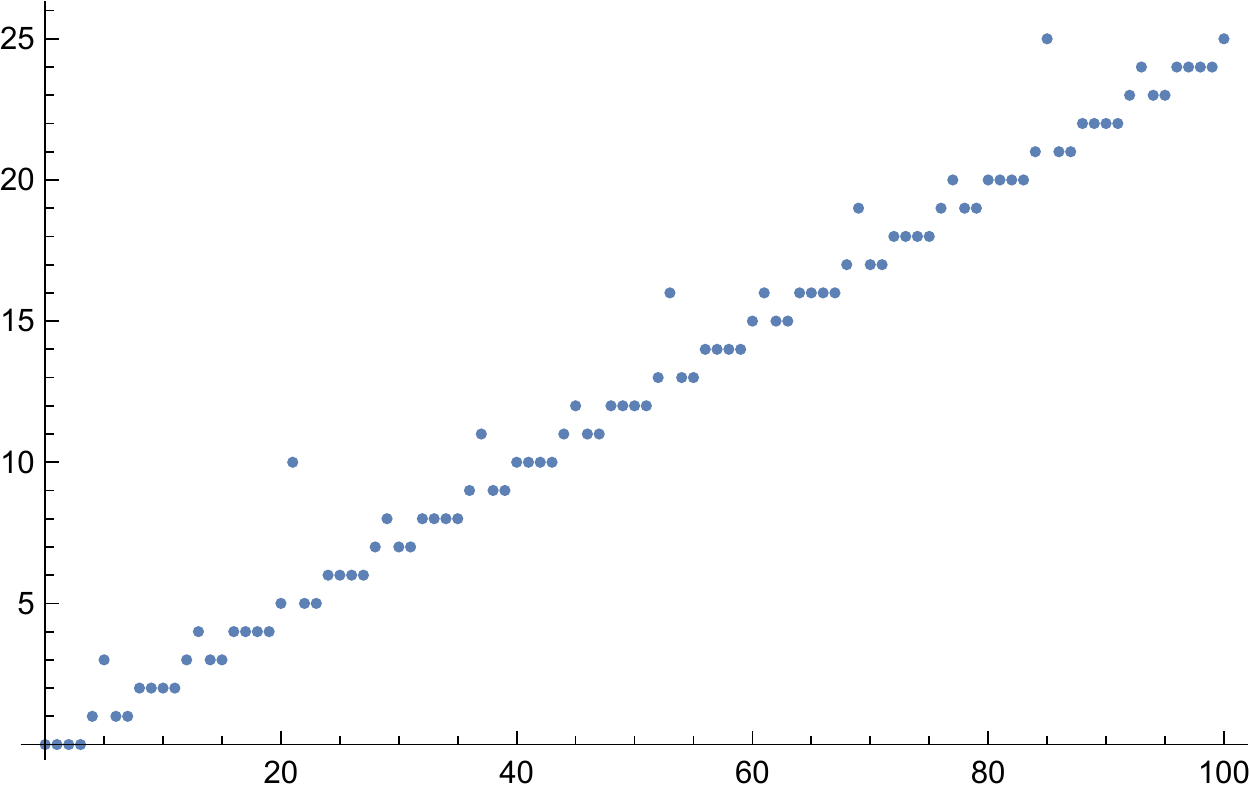}
	\caption{Plot of the $2$-adic valuation of the sequence in Example~\ref{exa:ppp}.}
	\label{fig: fourth-order p=2}
\end{figure}
Hence, one can conjecture that $\alpha_1=\frac{1}{4}$ and, assuming $\epsilon_1$ to be negligible, the above condition~\eqref{eq:test} (with $u=1$) becomes 
$$4>\log_{2.804}(2)\simeq 0.672.$$

Take the numeration system of Example~\ref{exa:toy}. 
If we compute $\nu_2(U_i)$, the values for $41\le i\le 60 $ are given by 
$$24, 20, 21, 21, 24, 22, 23, 23, 27, 24, 25, 25, 28, 26, 27, 27, 33, 
28, 29, 29$$ and, similarly, 
for $\nu_3(U_i)$
$$13, 14, 14, 14, 15, 15, 15, 16, 17, 16, 17, 17, 17, 18, 18, 18, 19, 
20, 19, 20.$$
These sequences are plotted in Figures~\ref{fig: third-order p=2} and \ref{fig: third-order p=3}.
\begin{figure}
	\includegraphics[width=.5\textwidth]{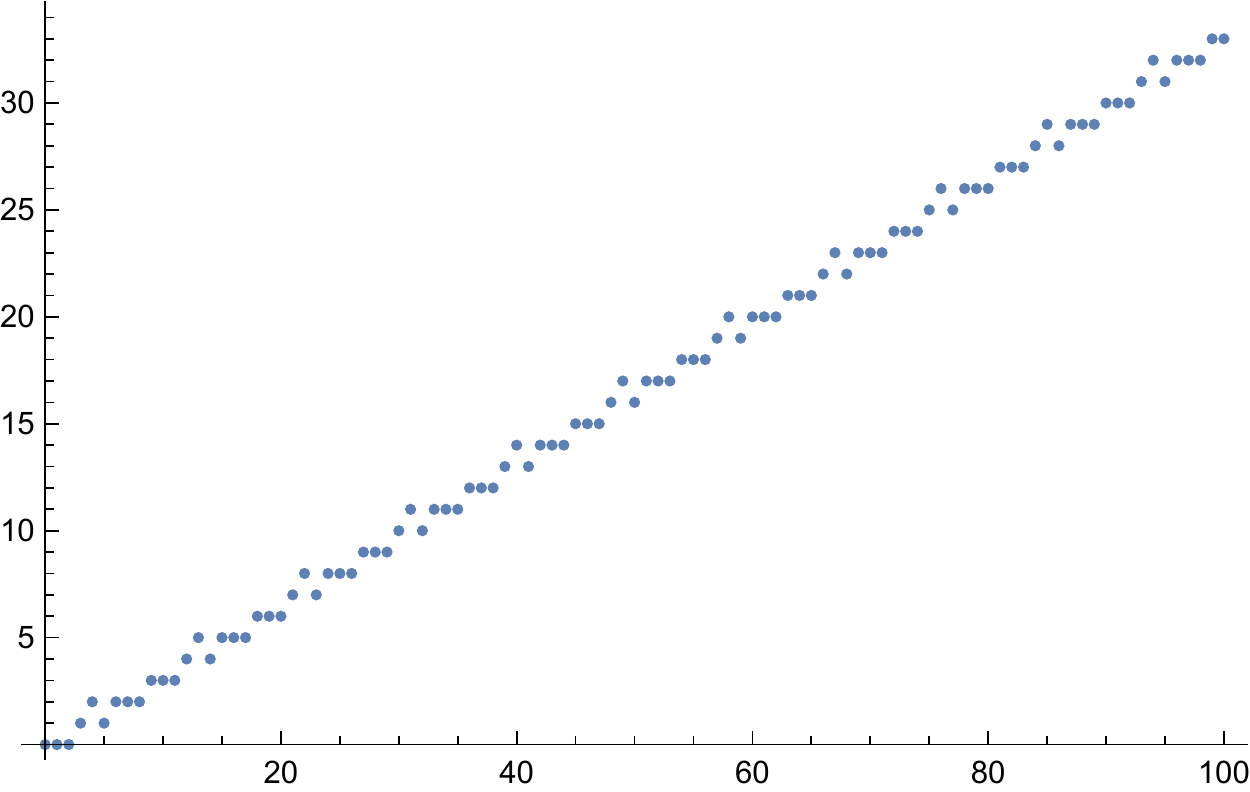}
	\caption{Plot of the $3$-adic valuation of the sequence in Example~\ref{exa:toy}.}
	\label{fig: third-order p=3}
\end{figure}
Hence, one can conjecture that $\alpha_1=\frac{1}{2}$ and $\alpha_2=\frac{1}{3}$. The recurrence has a real dominant root $\beta\simeq 12.554$. 
Assuming $\epsilon_1$ and  $\epsilon_2$ to be negligible, the condition~\eqref{eq:test} is therefore
$$2>\log_{12.554}(2)+\log_{12.554}(3)\simeq 0.708.$$
\end{remark}


\section{An incursion into $p$-adic analysis}\label{sec:ap}

In this section, we discuss the requirement on the $p$-adic valuation given in Lemma~\ref{lem:valp}.
We are able to show that this condition holds in certain cases.
In other cases, obtaining this condition requires information about the blocks of zeroes in the digit sequences of certain $p$-adic numbers, and in general it is not known how to get this information.

\subsection{A third-order sequence}

We reconsider our toy example.
Throughout this section, let $U_{i + 3} = 12 U_{i + 2} + 6 U_{i + 1} + 12 U_i$ with initial conditions $U_0 = 1, U_1 = 13, U_2 = 163$ be the sequence of Example~\ref{exa:toy}. The $3$-adic valuation of $U_i$, shown in Figure~\ref{fig: third-order p=3}, has a simple structure.

\begin{theorem}\label{p = 3 valuation}
For all $i \geq 0$,
\[
	\nu_3(U_i)
	= \left\lfloor\frac{i}{3}\right\rfloor +
	\begin{cases}
		1	& \text{if }i \equiv 4 \pmod 9 \\
		0	& \text{if }i \nequiv 4 \pmod 9.
	\end{cases}
\]
\end{theorem}

\begin{proof}
Let $T_i = U_i / 3^{\frac{i - 2}{3}}$.
Since $U_{i + 3} = 12 U_{i + 2} + 6 U_{i + 1} + 12 U_i$, the sequence $(T_i)_{i \geq 0}$ satisfies the recurrence $T_{i + 3} = 4 \cdot 3^{2/3} T_{i + 2} + 2 \cdot 3^{1/3} T_{i + 1} + 4 T_i$.
The initial terms are $T_0 = 3^{2/3}, T_1 = 13 \cdot 3^{1/3}, T_2 = 163$, so it follows that $T_i \in \Z[3^{1/3}]$ for all $i \geq 0$.
Modulo $9 \Z[3^{1/3}]$, one computes that the sequence $(T_i)_{i \geq 0}$ is periodic with period length $27$ and period
\[
\begin{array}{ccccccccc}
	3^{2/3}, & 4 \cdot 3^{1/3}, & 1, & 7 \cdot 3^{2/3}, & 3 \cdot 3^{1/3}, & 1, & 2 \cdot 3^{2/3}, & 2 \cdot 3^{1/3}, & 4, \\
	3^{2/3}, & \phantom{1 \cdot {}} 3^{1/3}, & 7, & 7 \cdot 3^{2/3}, & 3 \cdot 3^{1/3}, & 7, & 8 \cdot 3^{2/3}, & 5 \cdot 3^{1/3}, & 1, \\
	3^{2/3}, & 7 \cdot 3^{1/3}, & 4, & 7 \cdot 3^{2/3}, & 3 \cdot 3^{1/3}, & 4, & 5 \cdot 3^{2/3}, & 8 \cdot 3^{1/3}, & 7.
\end{array}
\]
Therefore the sequence $(\nu_3(T_i))_{i \geq 0}$ of $3$-adic valuations is 
$$\frac{2}{3},\  \frac{1}{3},\  0,\  \frac{2}{3},\  \frac{4}{3},\  0,\  \frac{2}{3},\  \frac{1}{3},\  0,\  \dots$$ with period length $9$.
(Here we use the natural extension of $\nu_3$ to a function $\nu_3 \colon \Z[3^{1/3}] \to \frac{1}{3} \Z$.)
Equivalently,
\[
	\nu_3(T_i)
	= \left\lfloor\frac{i}{3}\right\rfloor - \frac{i - 2}{3} +
	\begin{cases}
		1	& \text{if }i \equiv 4 \pmod 9 \\
		0	& \text{if }i \nequiv 4 \pmod 9.
	\end{cases}
\]
It follows that
\[
	\nu_3(U_i)
	= \frac{i - 2}{3} + \nu_3(T_i)
	= \left\lfloor\frac{i}{3}\right\rfloor +
	\begin{cases}
		1	& \text{if }i \equiv 4 \pmod 9 \\
		0	& \text{if }i \nequiv 4 \pmod 9
	\end{cases}
\]
for all $i \geq 0$.
\end{proof}

Theorem~\ref{p = 3 valuation} implies $\frac{i - 2}{3} \leq \nu_3(U_i) \leq \frac{i + 2}{3}$ for all $i \geq 0$. In particular, $\nu_3(U_i) < \lfloor\frac{i}{3}\rfloor + 2$, so the condition of Lemma~\ref{lem:valp} is satisfied, and therefore for every $\epsilon > 0$ we have
\[
	\mathsf{f}_3(\mu) > \frac{\mu}{\frac{1}{3} + \epsilon}
\]
for large enough $\mu$.
This takes care of one of the two primes dividing $6$, the gcd of the coefficients of the recurrence relation.
{ To apply Theorem~\ref{thm:main2}, it remains to bound $\nu_2(U_i)$.}

However, Theorem~\ref{p = 3 valuation} is not representative of the behavior of $\nu_p(s_i)$ for a general sequence $(s_i)_{i \geq 0}$ satisfying a linear recurrence with constant coefficients.
For instance, the $2$-adic valuation of $(U_i)_{i \geq 0}$ is (much) more complicated. To study the more general setting, we will make use of the field $\Q_p$ of $p$-adic numbers and its ring of integers $\Z_p$.
The $p$-adic valuation $\nu_p(x)$ of an element $x \in \Q_p$ is related to its $p$-adic absolute value $\size{x}_p$ by $\size{x}_p = p^{-\nu_p(x)}$.
For an introduction to $p$-adic analysis, see \cite{Gouvea}.

Let $\size{\rep_p(n)}$ be the number of digits in the standard base-$p$ representation of $n$.
For all $n \geq 1$, we can bound $\nu_p(n)$ as
\[
	\nu_p(n) \leq \size{\rep_p(n)} - 1 = \left\lfloor \tfrac{1}{\log(p)} \log(n) \right\rfloor \leq \tfrac{1}{\log(p)} \log(n).
\]
(We avoid writing ``$\log_p(n)$'' here to reserve $\log_p$ for the $p$-adic logarithm, which will come into play shortly.)
Proposition~\ref{log upper bound} below gives the analogous upper bound on $\nu_p(n - \zeta)$ when $\zeta$ is a $p$-adic integer whose sequence of base-$p$ digits does not have blocks of consecutive $0$s that grow too quickly.

\begin{notation*}
Let $p$ be a prime, and let $\zeta \in \Z_p \setminus \N$.
Write $\zeta = \sum_{i \geq 0} d_i p^i$, where each $d_i \in \{0, 1, \dots, p - 1\}$.
For each $a \geq 0$, let $\ell_\zeta(a) \geq 0$ be maximal such that $0 = d_a = d_{a + 1} = \dots = d_{a + \ell_\zeta(a) - 1}$.
\end{notation*}

\begin{prop}\label{log upper bound}
Let $p$ be a prime, and let $\zeta \in \Z_p \setminus \N$.
If there exist real numbers $C, D$ such that $C > 0$, $D \geq -(C + 1)$, and $\ell_\zeta(a) \leq C a + D$ for all $a \geq 2$, then $\nu_p(n - \zeta) \leq \frac{2 C + D + 2}{\log(p)} \log(n)$ for all $n \geq p$.
\end{prop}

\begin{proof}
Write $\zeta = \sum_{i \geq 0} d_i p^i$, where each $d_i \in \{0, 1, \dots, p - 1\}$.
For each $a \geq 0$, define the integer $\zeta_a \colonequal \left(\zeta \bmod p^a\right) = \sum_{i = 0}^{a - 1} d_i p^i$.
Then $\nu_p(\zeta_a - \zeta) = a + \ell_\zeta(a)$.

Let $n \geq p$, and let $a \colonequal \size{\rep_p(n)} \geq 2$.
Since $\zeta \notin \N$, the $p$-adic valuation $b \colonequal \nu_p(n - \zeta)$ is an integer.
There are two cases.

If $n \leq \zeta_b$, then in fact $n = \zeta_b$; this is because $n \leq \zeta_b < p^b$, so $n \neq \zeta_b$ implies $n - \zeta_b \nequiv 0 \pmod{p^b}$, which contradicts $b = \nu_p(n - \zeta)$.
Since $\size{\rep_p(n)} = a$ and $n = \zeta_b$, we have $0 = d_a = \dots = d_{b - 1}$.
Therefore $\zeta_a = \zeta_b = n \geq p^{a - 1}$, and
\[
	\frac{\nu_p(n - \zeta)}{\log(n)}
	= \frac{\nu_p(\zeta_a - \zeta)}{\log(\zeta_a)}
	\leq \frac{a + \ell_\zeta(a)}{\log(p^{a - 1})}
	\leq \frac{a + C a + D}{(a - 1) \log(p)}
	\leq \frac{2 + 2 C + D}{\log(p)},
\]
where the final inequality follows from $1 + C + D \geq 0$.

If $n > \zeta_b$, then $n = \zeta_b + p^b m$ for some positive integer $m$.
Therefore $n  \geq p^b$, so
\[
	\frac{\nu_p(n - \zeta)}{\log(n)}
	\leq \frac{b}{\log(p^b)}
	= \frac{1}{\log(p)}
	< \frac{1 + C}{\log(p)}
	\leq \frac{2 + 2 C + D}{\log(p)}
\]
if $b \geq 1$ and $\frac{\nu_p(n - \zeta)}{\log(n)} = 0 < \frac{2 + 2 C + D}{\log(p)}$ if $b = 0$.
\end{proof}

We now turn our attention to the sequence of $2$-adic valuations $\nu_2(U_i)$.
The following result concerns the local peaks in Figure~\ref{fig: third-order p=2}.

\begin{theorem}\label{p = 2 zero}
There exists a unique $2$-adic integer $\zeta$ with the property that if $(i_n)_{n \geq 0}$ is a sequence of non-negative integers such that {$\nu_2(U_{i_n}) - \frac{i_n}{2} \to \infty$} then $i_n \to \zeta$ in $\Z_2$.
\end{theorem}

A formula for $\zeta$ is given by Equation~\eqref{zeta definition} in the proof.
In particular, $\zeta$ is a computable number, and one computes $\zeta \equiv 660098850944665 \pmod{2^{50}}$.

\begin{proof}[Proof of Theorem~\ref{p = 2 zero}]
Let $p = 2$.
To analyze the $2$-adic behavior of $(U_i)_{i \geq 0}$, we construct a piecewise interpolation of $U_i$ to $\Z_2$ using the method described by Rowland and Yassawi~\cite{Rowland--Yassawi}.
Let $P(x) = x^3 - 12 x^2 - 6 x - 12$ be the characteristic polynomial of $(U_i)_{i \geq 0}$.
The polynomial $P(x)$ has a unique root $\beta_1 \in \Z_2$ satisfying $\beta_1 \equiv 2 \pmod 4$; this can be shown by an application of Hensel's lemma (checking $\size{P(2)}_2 < \size{P'(2)}_2^2$).
Polynomial division shows that $P(x)$ factors in $\Z_2[x]$ as
\[
	P(x) = (x - \beta_1) \left(x^2 + (\beta_1 - 12) x + (\beta_1^2 - 12 \beta_1 - 6)\right).
\]
One checks that $P(x)$ has no roots in $\Z_2$ congruent to $0$, $1$, $3$, $4$, $5$, or $7$ modulo $8$.
Since $\beta_1$ has multiplicity $1$, this implies that the splitting field $K$ of $P(x)$ is a quadratic extension of $\Q_2$.
Let $\beta_2$ and $\beta_3$ be the other two roots of $P(x)$ in $K = \Q_2(\beta_2)$.
Since $\beta_1 \equiv 2 \pmod 4$, the $2$-adic absolute value of $\beta_1$ is $\size{\beta_1}_2 = \frac{1}{2}$.
Using the quadratic factor of $P(x)$ and an approximation to $\beta_1$, one computes $\size{\beta_2}_2 = \size{\beta_3}_2 = \frac{1}{\sqrt{2}}$.

Let $c_1, c_2, c_3 \in K$ be such that
\[
	U_i = c_1 \beta_1^i + c_2 \beta_2^i + c_3 \beta_3^i
\]
for all $i \geq 0$.
Using the initial conditions, we solve for $c_1, c_2, c_3$ to find
\begin{align*}
	c_1 &= \frac{-U_0 \beta_2 \beta_3 + U_1 (\beta_2 + \beta_3) - U_2}{(\beta_2 - \beta_1) (\beta_1 - \beta_3)} \\
	c_2 &= \frac{-U_0 \beta_3 \beta_1 + U_1 (\beta_3 + \beta_1) - U_2}{(\beta_3 - \beta_2) (\beta_2 - \beta_1)} \\
	c_3 &= \frac{-U_0 \beta_1 \beta_2 + U_1 (\beta_1 + \beta_2) - U_2}{(\beta_1 - \beta_3) (\beta_3 - \beta_2)},
\end{align*}
where  $U_0 = 1, U_1 = 13, U_2 = 163$.
One computes $\size{c_1}_2 = 2$ and $\size{c_2}_2 = 2 \sqrt{2} = \size{c_3}_2$.
Factoring out $\beta_2^i$ gives
\begin{equation}\label{2-adic valuation}
	U_i
	= \beta_2^i \left(c_1 \, (\tfrac{\beta_1}{\beta_2})^i + c_2 + c_3 \, (\tfrac{\beta_3}{\beta_2})^i\right).
\end{equation}
Since $\size{\frac{\beta_1}{\beta_2}}_2 = \frac{1}{\sqrt{2}}$ and $\size{\frac{\beta_3}{\beta_2}}_2 = 1$, the power $(\tfrac{\beta_1}{\beta_2})^i$ approaches $0$ as $i \to \infty$, while $(\tfrac{\beta_3}{\beta_2})^i$ does not.
Therefore the size of $\nu_2(U_i / \beta_2^i) = \nu_2(U_i) - \frac{i}{2}$ for large $i$ is determined by the proximity of $c_2 + c_3 \, (\tfrac{\beta_3}{\beta_2})^i$ to $0$.

To analyze the size of $c_2 + c_3 \, (\tfrac{\beta_3}{\beta_2})^i$, we interpret $(\frac{\beta_3}{\beta_2})^i$ as a function of a $p$-adic variable.
For this we need the $p$-adic exponential and logarithm, which are defined on extensions of $\Q_p$ by their usual power series;
$\log_p(1 + x)$ converges if $\lvert x \rvert_p < 1$, and $\exp_p x$ converges if $\lvert x \rvert_p < p^{-1/(p - 1)}$.
Moreover, $\log_p$ is an isomorphism from the multiplicative group $\{x : \lvert x - 1 \rvert_p < p^{-1/(p - 1)}\}$ to the additive group $\{x : \lvert x \rvert_p < p^{-1/(p - 1)}\}$, and its inverse map is $\exp_p$~\cite[Proposition~4.5.9 and Section~6.1]{Gouvea}.
Direct computation shows $\size{(\frac{\beta_3}{\beta_2})^4 - 1}_2 = \frac{1}{8} < \frac{1}{2} = p^{-1/(p - 1)}$.
Therefore, for all $m \geq 0$ and $r \in \{0, 1, 2, 3\}$,
\begin{align*}
	(\tfrac{\beta_3}{\beta_2})^{r + 4 m}
	&= (\tfrac{\beta_3}{\beta_2})^r (\tfrac{\beta_3}{\beta_2})^{4 m}\\
	&= (\tfrac{\beta_3}{\beta_2})^r \exp_2 \log_2((\tfrac{\beta_3}{\beta_2})^{4 m}) \\
	&= (\tfrac{\beta_3}{\beta_2})^r \exp_2\!\left(m \log_2((\tfrac{\beta_3}{\beta_2})^4)\right).
\end{align*}
Denote $L \colonequal \log_2((\frac{\beta_3}{\beta_2})^4)$.
Using the power series for $\log_2$, one computes $\size{L}_2 = \frac{1}{8}$.
For each $x \in \Z_2[\beta_2]$ and $r \in \{0, 1, 2, 3\}$, define
\[
	f_r(r + 4 x) \colonequal c_2 + c_3 \, (\tfrac{\beta_3}{\beta_2})^r \exp_2(L x).
\]
For all $x \in \Z_2$, we have $\size{L x}_2 = \frac{1}{8} \size{x}_2 \leq \frac{1}{8} < \frac{1}{2} = p^{-1/(p - 1)}$, so $f_r$ is well defined on $r + 4 \Z_2$.
The four functions $f_0, f_1, f_2, f_3$ comprise a piecewise interpolation of $c_2 + c_3 \, (\frac{\beta_3}{\beta_2})^i$.
Namely, $c_2 + c_3 \, (\frac{\beta_3}{\beta_2})^i = f_{i \bmod 4}(i)$ for all $i \geq 0$.

Since each $f_r$ is a continuous function, from Equation~\eqref{2-adic valuation} we see that $\nu_2(U_i / \beta_2^i) = \nu_2(U_i) - \frac{i}{2}$ is large when $i$ is close to a zero of $f_{i \bmod 4}$.
The equation $f_r(r + 4 x) = 0$ is equivalent to
\[
	\exp_2(L x)
	= -\tfrac{c_2}{c_3} (\tfrac{\beta_2}{\beta_3})^r.
\]
For $r \in \{0, 2, 3\}$, one computes $\sizedsize{-\frac{c_2}{c_3} (\frac{\beta_2}{\beta_3})^r - 1}_2 \geq \frac{1}{2}$, so there is no solution $x$ for these values of $r$.
For $r = 1$, $\sizedsize{-\frac{c_2}{c_3} (\frac{\beta_2}{\beta_3})^r - 1}_2 = \frac{1}{16} < \frac{1}{2}$, so there is a unique solution, namely $x = \frac{1}{L} \log_2\!\left(-\frac{c_2 \beta_2}{c_3 \beta_3}\right)$, which has size $\size{x}_2 = \frac{1}{2}$.
Let
\begin{equation}\label{zeta definition}
	\zeta \colonequal 1 + 4 \tfrac{1}{L} \log_2\!\left(-\tfrac{c_2 \beta_2}{c_3 \beta_3}\right),
\end{equation}
so that $f_1(\zeta) = 0$ and $\size{\zeta}_2 = 1$.
It follows that every sequence $(i_n)_{n \geq 0}$ of non-negative integers with $\nu_2(U_{i_n}) - \frac{i_n}{2} \to \infty$ satisfies $i_n \to \zeta$.
(If $\zeta \notin \Z_2$, then such sequences do not exist.)

It remains to show that $\zeta \in \Z_2$.
Let $\sigma : K \to K$ be the Galois automorphism that non-trivially permutes $\beta_2$ and $\beta_3$.
The formulas for $c_2$ and $c_3$ imply $\frac{c_2}{c_3} \cdot \frac{\sigma(c_2)}{\sigma(c_3)} = 1$;
this implies
\begin{align*}
	\log_2\!\left(-\tfrac{c_2 \beta_2}{c_3 \beta_3}\right)
		+ \sigma\!\left(\log_2\!\left(-\tfrac{c_2 \beta_2}{c_3 \beta_3}\right)\right)
	&= \log_2\!\left(\tfrac{c_2 \beta_2}{c_3 \beta_3} \cdot \tfrac{\sigma(c_2) \beta_3}{\sigma(c_3) \beta_2}\right) \\
	&= \log_2(1)
	= 0.
\end{align*}
Similarly,
\[
	\log_2((\tfrac{\beta_3}{\beta_2})^4) + \sigma\!\left(\log_2((\tfrac{\beta_3}{\beta_2})^4)\right)
	= \log_2(1)
	= 0.
\]
Therefore
\[
	\frac{
		\log_2\!\left(-\tfrac{c_2 \beta_2}{c_3 \beta_3}\right)
	}{
		\log_2((\frac{\beta_3}{\beta_2})^4)
	}
	= \frac{
		-\sigma\!\left(\log_2\!\left(-\tfrac{c_2 \beta_2}{c_3 \beta_3}\right)\right)
	}{
		-\sigma\!\left(\log_2((\frac{\beta_3}{\beta_2})^4)\right)
	}
	= \sigma\!\left(\frac{
		\log_2\!\left(-\tfrac{c_2 \beta_2}{c_3 \beta_3}\right)
	}{
		\log_2((\frac{\beta_3}{\beta_2})^4)
	}\right)
\]
is invariant under $\sigma$ and thus is an element of $\Q_2$.
It follows from $\size{\zeta}_2 = 1$ that $\zeta \in \Z_2$.
\end{proof}

\begin{remark*}
The interpolation in the previous proof depends on appropriate powers of $\frac{\beta_3}{\beta_2}$ satisfying $x = \exp_2(\log_2(x))$.
We verified this by directly checking $\size{(\frac{\beta_3}{\beta_2})^4 - 1}_2 < \frac{1}{2}$.
In general, an appropriate exponent is given by \cite[Lemma~6]{Rowland--Yassawi}, namely
\[
	\begin{cases}
		1						& \text{if $e < p - 1$} \\
		p^{\lceil \log(e+1) / \log p \rceil}	& \text{if $e \geq p - 1$,}
	\end{cases}
\]
where $e$ is the ramification index of the field extension.
The ramification index of the extension $K$ in the proof of Theorem~\ref{p = 2 zero} is $e = 2$; this follows from the fact that $e$ is a divisor of the degree of the extension and that $e \neq 1$ since we identified an element $\beta_2 \in K$ with $2$-adic valuation $\nu_2(\beta_2) = \frac{1}{2}$.
Therefore the exponent $2^{\lceil \log(3) / \log(2) \rceil} = 4$ suffices.
Since $\size{\frac{\beta_3}{\beta_2}}_2 = 1$, \cite[Lemma~6]{Rowland--Yassawi} implies $\size{(\frac{\beta_3}{\beta_2})^4 - 1}_2 < \frac{1}{2}$.
(In general, one must divide by a root of unity before raising to the appropriate exponent, but this root of unity is $1$ for $\frac{\beta_3}{\beta_2}$ since the ramification index of $K$ is equal to its degree.)
\end{remark*}

By Proposition~\ref{log upper bound}, the growth rate of $\nu_2(U_i)$ is determined by the approximability of
\[
	\zeta = \cdots 10010110000101101100111101100001101111011010011001_2
\]
by non-negative integers.

\begin{conjecture}\label{p = 2 block lengths}
Let $\zeta \in \Z_2$ be defined as in Equation~\eqref{zeta definition}.
The lengths of the $0$ blocks of the $2$-adic digits of $\zeta$ satisfy $\ell_\zeta(a) \leq \frac{2}{95} a + \frac{18}{5}$ for all $a \geq 0$.
\end{conjecture}

Conjecture~\ref{p = 2 block lengths} is weak in the sense that it is almost certainly far from sharp.
One expects the digits of $\zeta$ to be randomly distributed, in which case $\ell_\zeta(a) = \frac{1}{\log(2)} \log(a) + O(1)$.
Indeed, among the first $1000$ base-$2$ digits of $\zeta$, the longest block of $0$s has length $10$.
However, results concerning digits of irrational numbers are notoriously difficult to prove.
Bugeaud and Keke\c{c}~\cite[Theorem~1.6]{Bugeaud--Kekec} give a lower bound on the number of non-zero digits among the first $a$ digits of an irrational algebraic number in $\Q_p$; see also Theorem~2.1 in the same paper.
However, there are no known results of this form for transcendental numbers.

The conjectural bound was obtained by computing the line through $\ell_\zeta(19) = 4$ and $\ell_\zeta(304) = 10$.
If Conjecture~\ref{p = 2 block lengths} is true, then an explicit formula for $\nu_2(U_i)$ is given by the following theorem.
In particular, the approximation $\zeta \equiv 660098850944665 \pmod {2^{50}}$ is sufficient to compute $\nu_2(U_i)$ for all $i \leq 2^{49}$.

\begin{theorem}\label{p = 2 valuation}
Let $\zeta \in \Z_2$ be defined as in Equation~\eqref{zeta definition}.
Conjecture~\ref{p = 2 block lengths} implies that, for all $i \geq 10$,
\[
	\nu_2(U_i)
	= \left\lfloor\frac{i - 1}{2}\right\rfloor +
	\begin{cases}
		\nu_2(i - \zeta)	& \text{if }i \equiv 1 \pmod 4 \\
		0			& \text{if }i \nequiv 1 \pmod 4.
	\end{cases}
\]
\end{theorem}

\begin{proof}
We start as in the proof of Theorem~\ref{p = 3 valuation}.
Let $T_i = U_i / 2^{\frac{i}{2} - 1}$.
Since $U_{i + 3} = 12 U_{i + 2} + 6 U_{i + 1} + 12 U_i$, the sequence $(T_i)_{i \geq 0}$ satisfies the recurrence $T_{i + 3} = 6 \sqrt{2} T_{i + 2} + 3 T_{i + 1} + 3 \sqrt{2} T_i$.
The initial terms are $T_0 = 2, T_1 = 13 \sqrt{2}, T_2 = 163$, so it follows that $T_i \in \Z[\sqrt{2}]$ for all $i \geq 0$.
Modulo $2 \Z[\sqrt{2}]$, the sequence $(T_i)_{i \geq 2}$ is periodic with period length $4$: $1, \sqrt{2}, 1, 0, 1, \sqrt{2}, 1, 0, \dots$.
It follows that if $i \geq 2$ and $i \nequiv 1 \pmod 4$ then
\begin{align*}
	\nu_2(U_i)
	= \frac{i}{2} - 1 + \nu_2(T_i)
	&= \frac{i}{2} - 1 +
	\begin{cases}
		0		& \text{if }i \equiv 0 \pmod 4 \\
		0		& \text{if }i \equiv 2 \pmod 4 \\
		\frac{1}{2}	& \text{if }i \equiv 3 \pmod 4
	\end{cases} \\
	&= \left\lfloor\frac{i - 1}{2}\right\rfloor.
\end{align*}

It remains to determine $\nu_2(U_i)$ when $i \equiv 1 \pmod 4$.
We continue to use the $2$-adic numbers $\beta_1, \beta_2, \beta_3, c_1, c_2, c_3$ and the function $f_1$ defined in the proof of Theorem~\ref{p = 2 zero}.
When $i \equiv 1 \pmod 4$, Equation~\eqref{2-adic valuation} gives
\[
	\sizedsize{U_i}_2
	= 2^{-\frac{i}{2}} \sizedsize{c_1 \, (\tfrac{\beta_1}{\beta_2})^i + f_1(i)}_2.
\]
To obtain a simpler formula for $\sizedsize{U_i}_2$, we compare the sizes of the two terms being added and use the fact that $\size{x + y}_p = \max\{\size{x}_p, \size{y}_p\}$ if $\size{x}_p \neq \size{y}_p$.
For the first, we have $\sizedsize{c_1 \, (\tfrac{\beta_1}{\beta_2})^i}_2 = 2^{1 - \frac{i}{2}}$.
For the second,
\[
	\sizedsize{f_1(i)}_2
	= \sizedsize{c_2 + \tfrac{c_3 \beta_3}{\beta_2} \exp_2\!\left(L \cdot \tfrac{i - 1}{4}\right)}_2.
\]
Since the function $f_1(1 + 4 x) = c_2 + \tfrac{c_3 \beta_3}{\beta_2} \exp_2(L x)$ has a unique zero $\frac{\zeta - 1}{4}$, the $p$-adic Weierstrass preparation theorem~\cite[Theorem~6.2.6]{Gouvea} implies the existence of a power series $h(x) \in K\llbracket x \rrbracket$ such that $h(0) = 1$, $\size{h(x)}_2 = 1$ for all $x \in \Z_2[\beta_2]$, and
\[
	f_1(1 + 4 x)
	= \frac{c_2 + \tfrac{c_3 \beta_3}{\beta_2}}{-\tfrac{\zeta - 1}{4}} \left(x - \tfrac{\zeta - 1}{4}\right) h(x).
\]
Therefore
\begin{align*}
	\sizedsize{f_1(i)}_2
	&= \sizedsize{\frac{c_2 + \tfrac{c_3 \beta_3}{\beta_2}}{-\tfrac{\zeta - 1}{4}}}_2
		\sizedsize{\tfrac{i - 1}{4} - \tfrac{\zeta - 1}{4}}_2 \\
	&= \sqrt{2} \sizedsize{i - \zeta}_2.
\end{align*}
Conjecture~\ref{p = 2 block lengths} and Proposition~\ref{log upper bound} imply $\size{i - \zeta}_2 \geq i^{-536/95}$ for all $i \geq 2$.
The functions $2^{1 - \frac{i}{2}}$ and $\sqrt{2} i^{-536/95}$ intersect at $i \approx 70.21$.
For all $i \geq 73$ such that $i \equiv 1 \pmod 4$,
\[
	\sizedsize{c_1 \, (\tfrac{\beta_1}{\beta_2})^i}_2
	= 2^{1 - \frac{i}{2}}
	< \sqrt{2} i^{-536/95}
	\leq \sizedsize{f_1(i)}_2
\]
and therefore
\[
	\sizedsize{U_i}_2
	= 2^{-\frac{i}{2}} \sizedsize{c_1 \, (\tfrac{\beta_1}{\beta_2})^i + f_1(i)}_2
 	= 2^{-\frac{i}{2}} \sizedsize{f_1(i)}_2
	= 2^{\frac{1 - i}{2}} \sizedsize{i - \zeta}_2.
\]
Moreover, explicit computation shows that $2^{1 - \frac{i}{2}} < \sqrt{2} \size{i - \zeta}_2$ for all $i \equiv 1 \pmod 4$ satisfying $13 \leq i \leq 69$, so $\sizedsize{U_i}_2 = 2^{\frac{1 - i}{2}} \sizedsize{i - \zeta}_2$ for these values as well.
Therefore $\nu_2(U_i) = \frac{i - 1}{2} + \nu_2(i - \zeta)$ for all $i \geq 13$ such that $i \equiv 1 \pmod 4$.
\end{proof}

\begin{corollary}
Conjecture~\ref{p = 2 block lengths} implies that $\nu_2(U_i) \leq \frac{i}{2} + \frac{536}{95 \log(2)} \log(i)$ for all $i \geq 10$.
\end{corollary}

\begin{proof}
Since $U_i \neq 0$ for all $i \geq 0$, we have $\sizedsize{U_i}_2 \neq 0$ for all $i \geq 0$.
Since $\sizedsize{f_1(\zeta)}_2 = 0$, this implies $\zeta \notin \N$.
Conjecture~\ref{p = 2 block lengths} and Proposition~\ref{log upper bound} imply $\nu_2(i - \zeta) \leq \frac{536}{95 \log(2)} \log(i)$ for all $i \geq 2$.
By Theorem~\ref{p = 2 valuation}, $\nu_2(U_i) \leq \frac{i}{2} + \frac{536}{95 \log(2)} \log(i)$ for all $i \geq 10$.
\end{proof}

This is sufficient to apply Lemma~\ref{lem:valp}. Assuming Conjecture~\ref{p = 2 block lengths},  we have the right behavior for both $\nu_2(U_i)$ and $\nu_3(U_i)$, { and therefore we may apply the decision procedure of Theorem~\ref{thm:main2}.}

\subsection{A fourth-order sequence}

Bounding the $p$-adic valuation of a sequence satisfying a recurrence of higher order is even more complicated than the proof of Theorem~\ref{p = 2 zero}.
For example, let $p = 2$ and consider the sequence $(U_i)_{i \geq 0}$ satisfying the recurrence $U_{i + 4} = 2 U_{i + 3} + 2 U_{i + 2} + 2 U_n$ with initial conditions $U_0 = 1, U_1 = 3, U_2 = 9, U_2 = 23$ from Example~\ref{exa:ppp}.
The $2$-adic valuation is shown in Figure~\ref{fig: fourth-order p=2}.
By the Eisenstein criterion, the characteristic polynomial $P(x) = x^4 - 2 x^3 - 2 x^2 - 2$ is irreducible over $\Q_2$.
Let $K$ be the splitting field of $P(x)$ over $\Q_2$.
Let $\beta_1, \beta_2, \beta_3, \beta_4$ be the four roots of $P(x)$ in $K$, and let $c_1, c_2, c_3, c_4$ be the elements of $K$ such that $U_i = \sum_{j = 1}^4 c_j \beta_j^i$ for all $i \geq 0$.

To compute with the roots $\beta_i$, we would want to write $K$ as a simple extension $\Q_2(\alpha)$.
For this, we need to determine the degree $d$ of the extension and a polynomial $Q(x) \in \Q_2[x]$ of degree $d$ such that $Q(x)$ is irreducible over $\Q_2$ and $Q(\alpha) = 0$. 
Then we could compare the sizes $\size{\beta_j}_2$ of the roots to each other.
Experiments suggest that $\size{\beta_1}_2 = \size{\beta_2}_2 = \size{\beta_3}_2 = \size{\beta_4}_2 = 2^{-1/4}$ and $\size{(\frac{\beta_j}{\beta_1})^8 - 1}_2 = \frac{1}{4} < \frac{1}{2} = p^{-1/(p - 1)}$ for each $j \in \{2, 3, 4\}$.
Assuming this is the case, $U_i / \beta_1^i = \sum_{j = 1}^4 c_j (\frac{\beta_j}{\beta_1})^i$ can be interpolated piecewise to $\Z_2$ using $8$ analytic functions.
However, we cannot solve $c_1 + b_2 \exp_2(L_2 x) + b_3 \exp_2(L_3 x) + b_4 \exp_2(L_4 x) = 0$ explicitly, as we solved $c_2 + c_3 \, (\tfrac{\beta_3}{\beta_2})^r \exp_2(L x) = 0$ in the proof of Theorem~\ref{p = 2 zero}.
Instead, we could use the $p$-adic Weierstrass preparation theorem~\cite[Theorem~6.2.6]{Gouvea} to determine the number of solutions and compute approximations to them.
However, we would also need to determine which of these solutions belong to $\Z_2$.
We do not carry out this step here, but this would give an analogue of Theorem~\ref{p = 2 zero}, with some finite set $Z$ of $2$-adic integers such that every sequence $(i_n)_{n \geq 0}$ of non-negative integers with $\nu_2(U_{i_n}) - \frac{i_n}{4} \to \infty$ satisfies $i_n \to \zeta$ for some $\zeta \in Z$.
If the blocks of zeroes in the digit sequences of each $\zeta \in Z$ satisfy $\ell_\zeta(a) \leq C a + D$ for some $C, D$ as in Conjecture~\ref{p = 2 block lengths}, then Proposition~\ref{log upper bound} gives an upper bound on $\nu_2(U_i)$.
This same approach applies to a general constant-recursive sequence and a general prime $p$.


\section{Concluding remarks}

The case of integer base~$b$ numeration systems is not treated in this paper. Let $b\ge 2$. Assume first for the sake of simplicity that $b$ is a prime. Consider the sequence $U=(b^i)_{i\ge 0}$. If $X$ is an ultimately periodic set with period $\pi_X=b^\lambda$ for some $\lambda$, then with our notation $Q_X=1$ and $\size{\rep_U(\pi_X-1)}=\lambda$. The sequence $(b^i \bmod{b^\lambda})_{i\ge 0}$ has a zero period and $\mathsf{f}_b(\lambda)=\lambda$. Hence we don't have the required assumption to apply Theorem~\ref{the:main}: for every such set $X$, $n_X=0$. Let us also point out that the technique of Proposition \ref{pro:cas2a} cannot be applied: adding $1$ as a most significant digit will not change the value of a representation modulo $\pi_X$ when words are too long, $U_i\equiv 0\pmod{b^\lambda}$ for large enough $i$. Of course, integer base systems can be handled with other decision procedures \cite{BMMR,BH,Honkala,LRRV,Marsault2019,MSaka}. If the base~$b$ is now a composite number of the form $p_1^{s_1}\cdots p_t^{s_t}$, the same observation holds. The length of the non-zero preperiod of $(b^i\bmod p_j^\mu)_{i\ge 0}$ is $\lfloor \frac{\mu}{s_j}\rfloor$. Taking again an ultimately periodic set with period $\pi_X=b^\lambda$, we get $Q_X=1$ and $\mathsf{f}_{p_j}(\lambda s_j)=\lambda$, hence $M_{\bm{\mu},X}=\lambda$ and we still have  $\size{\rep_U(\pi_X-1)}=\lambda$, so $n_X=0$.

A similar situation occurs in a slightly more general setting: the merge of $r$ sequences that ultimately behave like $b^i$. Let $b\ge 2$, $u\ge 1$, $N\ge 0$. If the recurrence relation is of the form $U_{i+u}=b U_i$ for $i\ge N$ (as for instance in Example~\ref{exa:noth2-newh2ok}), then again $n_X\not\to \infty$ as $\pi_X\to\infty$.
Indeed, if $X$ is an ultimately periodic set with period $\pi_X=b^\lambda$, then $Q_X=1$ and applying Lemma~\ref{lem:length} (here the polynomial $P_T$ with the notation of Definition~\ref{def:uT} is just a constant), $\size{\rep_U(\pi_X-1)}\ge u\lambda-L$, for some constant $L$, and with the same reasoning as for a composite integer base, $M_{\bm{\mu},X}\le N+u\lambda$. Thus $n_X$ remains bounded  for all $\lambda$. So there is no way to ensure that $n_X$ can be {larger than~$Z$}.

Trying to figure out the limitations of our decision procedure and assuming that we are under the assumption of Lemma~\ref{lem:valp}, this type of linear numeration systems is the only one that we were able to find where our procedure cannot be applied. Moreover, as shown by the following proposition, these systems are sufficiently close to the classical base-$b$ system so usual decision procedures can still be applied. It is an open problem to determine if there exist linear numeration systems satisfying (H1), (H2) and (H3) where the decision procedure may not be applied and not of the above type.

\begin{example}
   Take $b=4$, $u=2$ and $N=0$. Start with the first two values $1$ and $3$. We get the sequence  $1,3,4,12,16,48,64,\ldots$. We have $\mathsf{f}_2(\mu)=\mu$ if $\mu$ is even and $\mathsf{f}_2(\mu)=\mu+1$ if $\mu$ is odd. Hence, for a set of period $\pi_X=4^\lambda$, $M_{\bm{\mu},X}=\mathsf{f}_2(2\lambda)=2\lambda$. Moreover, $\size{\rep_U(4^\lambda-1)}=2\lambda$. So, $n_X=0$ for all $\lambda$.
\end{example}

\begin{prop}
Let $b\ge 2$, $u\ge 1$, $N\ge 0$. Let $U$ be a linear numeration system $U=(U_i)_{i\ge 0}$ such that $U_{i+u}=b U_i$ for all $i\ge N$. If a set is $U$-recognizable then it is $b$-recognizable. Moreover, given a DFA accepting $\rep_U(X)$ for some set $X$, we can compute a DFA accepting $\rep_b(X)$.
\end{prop}

\begin{proof}
   We build in two steps a sequence of transducers reading least significant digit first that maps any $U$-representation $c_{\ell-1}\cdots c_1 c_0\in A_U^*$ (here written with the usual convention that the most significant digit is on the left) to the corresponding $b$-ary representation. Adding leading zeroes, we may assume that the length $\ell$ of the $U$-representation is of the form $N+mu$. 
The idea is to read the first $N+u$ (least significant) digits and to output a single digit (over a finite alphabet in $\mathbb{N}$) equal to 
$$d_0=\val_U(c_{N+u-1}\cdots c_0).$$
Then we process blocks of size $u$, each such block of the form 
$$c_{N+(j+1)u-1}\cdots c_{N+ju}$$ gives as output a single digit  equal to 
$$d_j=c_{N+(j+1)u-1}U_{N+u-1}+\cdots +c_{N+ju} U_N.$$
So the digits $d_0,d_1,\ldots,d_{m-1}$ all belong to the finite set
$$\{\val_U(w)\colon w\in A_U^* \text{ and } |w|\le N+u\}.$$
From the form of the recurrence, we have
$$\val_U(c_{N+mu-1}\cdots c_0)=\sum_{j=0}^{m-1} d_j b^j
=\val_b(d_{m-1}\cdots d_0).$$
So this transducer $\mathcal{T}$ maps any $U$-representation to a non-classical $b$-ary representation of the same integer. Precisely, when a DFA accepting $\rep_U(X)$ is given, we build a DFA accepting the language $$L=0^*\rep_U(X)\cap\{w\in A_U^*\colon |w|\equiv N\pmod u,\ |w|\ge N\}.$$
 Recall that if $L$ is a regular language then its image $\mathcal{T}(L)$ by a transducer is again regular. Moreover, $\val_b(\mathcal{T}(L))=X$.

Then, it is a classical result that normalization in base~$b$, i.e.\ mapping a representation over a non-canonical finite set of digits to the canonical expansion over $\{0,\ldots,b-1\}$ can be achieved by a transducer $\mathcal{N}$ \cite{Frou} (or \cite[p.~104]{FLANS}). To conclude with the proof, we compose these two transducers and consider the image $\mathcal{N}(0^*\mathcal{T}(L))=0^*\rep_b(X)$.
\end{proof}

With the above proposition, the decision problem for the merge of sequences ultimately behaving like $b^i$ (such as the numeration systems of Examples~\ref{exa:noth2-newh2ok} and~\ref{exa:merge}) can be reduced to the usual decision problem for integer bases.

\section*{Acknowledgments}

We thank Yann Bugeaud for pointing out relevant theorems in \cite{Bugeaud--Kekec}. We thank Juha Honkala, Victor Marsault (who served as external reviewers for \cite{PhD}) and the anonymous referee for their careful reading and their suggestions leading to many improvements along the text. We also thank Jo{\"e}l Ouaknine for pointing out \cite{Ou3} and the reference to the absolutely divergent problem.

\end{document}